\documentclass[man,floatsintext,a4paper]{apa7}

\usepackage[british]{babel}
\usepackage{csquotes}
\usepackage{amsmath,amssymb,amsthm,mathtools}
\usepackage{cancel}
\usepackage{bm}
\usepackage{graphicx}
\usepackage{booktabs}
\usepackage{algorithm}
\usepackage{algpseudocode}
\usepackage{xcolor}
\usepackage{hyperref}
\usepackage{subcaption}
\usepackage{caption}
\usepackage{tikz}
\usetikzlibrary{positioning,arrows.meta,fit,calc,backgrounds}
\tikzset{
  semnode/.style = {draw, minimum height=7mm, font=\small},
  obs/.style     = {semnode, rectangle, minimum width=8mm, fill=black!8},
  lat/.style     = {semnode, circle, minimum size=10mm, fill=black!8},
  given/.style   = {semnode, rectangle, minimum width=8mm, dashed},
  lbl/.style     = {font=\scriptsize, inner sep=1.5pt}
}

\hypersetup{
  colorlinks=true,
  linkcolor=blue!50!black,
  citecolor=blue!50!black,
  urlcolor=blue!50!black
}

\usepackage[style=apa,sortcites=true,sorting=nyt,backend=biber,natbib]{biblatex}
\DeclareLanguageMapping{british}{british-apa}
\DefineBibliographyExtras{british}{}
\AtEveryBibitem{\clearfield{annotation}}

\newtheorem{proposition}{Proposition}
\newtheorem{theorem}[proposition]{Theorem}
\newtheorem{lemma}[proposition]{Lemma}
\newtheorem{corollary}[proposition]{Corollary}

\newcommand{\tr}{\operatorname{tr}}
\newcommand{\elpd}{\operatorname{elpd}}
\newcommand{\lppd}{\operatorname{lppd}}
\newcommand{\loco}{\textsc{loco}}
\newcommand{\loso}{\textsc{loso}}
\newcommand{\E}{\mathbb{E}}
\newcommand{\Var}{\operatorname{Var}}
\newcommand{\R}{\mathbb{R}}
\newcommand{\one}{\mathbf{1}}

\newcommand{\rank}{\operatorname{rank}}

\title{Deterministic Leave-One-Cluster-Out Cross-Validation for Multilevel Bayesian Structural Equation Models}
\shorttitle{Deterministic LOCO for Multilevel SEM}

\authorsnames[1,{1,2},1,1]{Mohammad Alhyari, Haziq Jamil, Hans Montcho, Håvard Rue}
\authorsaffiliations{%
  {King Abdullah University of Science and Technology},
  {Universiti Brunei Darussalam}
}

\authornote{%
\addORCIDlink{Mohammad Alhyari}{0009-0009-8770-4660}
  \addORCIDlink{Haziq Jamil}{0000-0003-3298-1010},
  \addORCIDlink{Hans Montcho}{0000-0003-2510-2102},
  \addORCIDlink{Håvard Rue}{0000-0002-0222-1881}.
  
  Correspondence concerning this article should be addressed to Haziq Jamil, Computer, Electrical and Mathematical Sciences and Engineering (CEMSE) Division, King Abdullah University of Science and Technology, Thuwal, 23955-6900, Kingdom of Saudi Arabia. E-mail: haziq.jamil@kaust.edu.sa.
}

\abstract{%
We introduce a closed-form, refit-free procedure for leave-one-cluster-out (\loco{}) cross-validation in multilevel Gaussian Bayesian structural equation models (SEMs), together with predictive scoring of every nested submodel.
Conditional independence of clusters given the parameters expresses the cluster-deleted posterior as a functional of the full posterior.
The \loco{} predictive density is then a harmonic mean of the cluster likelihood, computable from a single fit in \texttt{INLAvaan}, the integrated nested Laplace approximation package for Bayesian SEM.
We evaluate the harmonic-mean expectation in closed form through a fully exponential Laplace approximation of the reciprocal cluster likelihood under a Gaussian posterior; candidate structural restrictions follow by Gaussian conditioning of the same Laplace summary.
The resulting Taylor elpd (expected log predictive density) scores are fully deterministic, requiring neither refitting nor Monte Carlo sampling, so the variance pathology of the naive harmonic-mean estimator does not arise.
A direct-sum decomposition of the compound-symmetric cluster covariance makes the cost independent of cluster size, orders of magnitude below brute-force refitting.
We validate against brute-force refits and Markov chain Monte Carlo in simulation, and illustrate the procedure on school-safety climate in PISA 2022 and on 16 candidate structures linking personality and well-being in the MIDUS sibling sample.
}


\usepackage{setspace}
\begin{document}
\maketitle

\section{Introduction}
\label{sec:intro}

Bayesian structural equation modelling (SEM) has matured into a standard tool in psychometric practice, with implementations spanning the major Markov chain Monte Carlo (MCMC) samplers \citep{muthen2012bayesian} and, more recently, the integrated nested Laplace approximation (INLA) framework \citep{rue2009approximate, vanniekerk2023new}.
This machinery has since been carried into SEMs with normally distributed outcomes \citep{jamil2026approximate}, where the latent variables are integrated out analytically.
Inference begins with a mode-finding step on the model parameters and a Laplace approximation about that mode, which further corrections refine into skewness-adjusted posterior marginals.
That work is implemented in the \texttt{INLAvaan} package \citep{jamil2026implementation}, which takes models in familiar \texttt{lavaan} syntax \citep{rosseel2012lavaan} and returns accurate posterior summaries in a fraction of the time.

With this widespread adoption, the substantive questions tasked to Bayesian SEM have scaled significantly in complexity.
Rather than choosing between a small set of theoretically motivated competitors, researchers increasingly navigate large spaces of candidate cross-loadings, structural regressions, and latent hierarchies \citep{jacobucci2016regularized,lu2016bayesian}.
\emph{Structure selection}, deciding which paths to include, is variable selection carried out on the latent structure, since a freed path admits a predictor to a structural equation. 
It is no longer a post-hoc check but a central modelling step, with direct implications for construct validity and theoretical parsimony.

The natural Bayesian tool for this selection step is predictive cross-validation.
Leave-one-out (LOO) estimates of the expected log predictive density (elpd) are by now standard practice in the Bayesian workflow \citep{gelman2026bayesian,vehtari2017practical}, providing a criterion grounded in out-of-sample generalisation rather than in-sample fit.
For SEM in particular, predictive criteria are an attractive alternative to traditional fit indices such as the root mean square error of approximation (RMSEA), the comparative fit index (CFI), and the Tucker--Lewis index (TLI), whose theoretical grounding assumes independent observations and large samples and whose behaviour under hierarchical dependence is at best opaque \autocite{ryu2009levelspecific,hsu2015detecting}.
Predictive criteria carry one important nuance in the multilevel case.
As \textcite{merkle2019bayesian} make explicit, observation-level LOO for clustered data answers a \emph{conditional} prediction question: predicting a new subject within an already-observed cluster, with the cluster's posterior random effect supplying borrowed strength.
The natural target for generalisation across populations is instead the \emph{marginal} predictive density of an entirely new cluster, with the cluster effect integrated against its prior, the marginal latent distribution, rather than against its data-conditional posterior.
The two targets differ both numerically and in what they reward.
Conditional leave-one-subject-out (\loso{}) favours models that fit cluster-specific idiosyncrasies, while marginal leave-one-cluster-out (\loco{}) favours models that capture stable population-level structure.
For structure selection, the goal is generalisation across populations, not interpolation within an observed cluster, and hence \loco{} is the appropriate target.

Software practice took some time to adapt, for reasons that were computational rather than conceptual.
Deleting a cluster deletes the latent effects attached to it, so the density of the held-out block is not the conditional one a sampler already evaluates but a marginal one, obtained by integrating those effects out \autocite{merkle2019bayesian}.
Where that integral has no closed form it must be computed for every cluster at every draw, whereas the observation-level score conditions on the latent effects and requires no integral at all, and \loso{} became the default even for multilevel models (see Section~\ref{sec:loso-conditional}).
Analytic marginalisation removes that obstacle.
\texttt{blavaan}'s Stan implementation integrates the latent variables out in closed form for continuous outcomes \autocite{merkle2021efficient}, so a cluster-level marginal likelihood comes as a by-product of the fit.
From there the score is routine, since the \texttt{loo} package \autocite{vehtari2017practical} requires only that density evaluated at every draw for every cluster, and Pareto-smoothed importance sampling (PSIS) applied to such a matrix returns a \loco{} score from a single fit, with no refits.
However, this route requires a full MCMC posterior sample and returns an importance-sampling estimate whose accuracy degrades precisely where a held-out cluster is influential enough to matter for model choice \autocite{merkle2026nuances}.

This paper supplies the required \loco~target by other means. We show that the \loco{} predictive density for a multilevel Gaussian Bayesian SEM admits a closed-form evaluation requiring no refits, no posterior sampling, and no importance weights.
The construction rests on three observations:

\begin{enumerate}
\item The cluster-level predictive density $p(\mathbf{y}_j \mid \bm\vartheta)$ of a cluster's responses on $p$ indicators, obtained by integrating out the within- and between-cluster latent effects, is a multivariate Gaussian with compound-symmetric covariance.
An explicit direct-sum decomposition into within- and between-cluster subspaces reduces all log-determinants, inverses, and quadratic-form evaluations to two $p \times p$ factorisations, at a cost independent of cluster size \autocite{rosseel2021evaluating}.

\item The cluster-deleted posterior satisfies the reweighting identity $\pi(\bm\vartheta \mid \mathbf{y}_{-j}) \propto \pi(\bm\vartheta \mid \mathbf{y}) / p(\mathbf{y}_j \mid \bm\vartheta)$, a direct consequence of conditional independence of clusters given $\bm\vartheta$.
The \loco{} predictive density then collapses to a multivariate harmonic-mean cluster-level \emph{conditional predictive ordinate} \citep[CPO;][]{gelfand1992model} computable from quantities that an encompassing Laplace-approximated fit in \texttt{INLAvaan} already produces.

\item The harmonic-mean expectation itself admits a fully exponential Laplace approximation \parencite{tierney1986accurate}, replacing a variance-pathological Monte Carlo estimator with a deterministic value of zero variance and controlled bias.
First- and second-order forms are available, differing in asymptotic accuracy.
\end{enumerate}

The same Laplace summary supports structure selection.
Candidate single-path restrictions become linear constraints, and the restricted Laplace summary follows from closed-form Gaussian conditioning, a rank-one update when a single path is restricted.
Every submodel nested within the encompassing model is therefore scored from the same fit (Section~\ref{sec:restrictions}), at a cost far below refitting each one in turn.
The approximations developed here are efficient enough to enumerate a moderately sized structural hypothesis space exhaustively, rather than search it.

\section{Multilevel Structural Equation Models and the Prediction Target}
\label{sec:setting}

We work throughout with continuous indicators and Gaussian latent variables, the regime in which the closed-form derivations below hold exactly.
A multilevel SEM is assembled from a measurement equation and a structural equation written twice, once for the within-cluster component and once for the between-cluster component.
We therefore fix notation first for the single-level model which supplies that pair, and build the two-level version from it.
Section~\ref{sec:loso-main} describes the \loso{} target as a special case, appropriate for predicting a new subject in an observed cluster.

For subjects $s = 1, \dots, N$ and indicators $i = 1, \dots, p$, the measurement model relates a $p$-vector of observed responses $\mathbf{y}_s \in \R^p$ to a $q$-vector of latent constructs $\bm\eta_s \in \R^q$ via
\begin{equation}
\mathbf{y}_s = \bm\nu + \bm\Lambda \bm\eta_s + \bm\epsilon_s,
\qquad \bm\epsilon_s \sim \mathcal{N}_p(0, \bm\Theta),
\label{eq:meas-single}
\end{equation}
with intercept $\bm\nu \in \R^p$, factor loadings $\bm\Lambda \in \R^{p \times q}$, and residual covariance $\bm\Theta \in \R^{p \times p}$ (typically diagonal).
The structural model relates latent constructs to one another and to a constant:
\begin{equation}
\bm\eta_s = \bm\alpha + \mathbf{B} \bm\eta_s + \bm\zeta_s,
\qquad \bm\zeta_s \sim \mathcal{N}_q(0, \bm\Psi),
\label{eq:struc-single}
\end{equation}
with $\bm\alpha \in \R^q$, $\mathbf{B} \in \R^{q \times q}$ the matrix of structural regression coefficients (the object of structure selection) and $\bm\Psi \in \R^{q \times q}$ the latent disturbance covariance.
We collect the SEM parameters as $\mathbb R^d \ni \bm\vartheta = \{\bm\nu, \bm\Lambda, \bm\Theta, \bm\alpha, \mathbf{B}, \bm\Psi\}$ and write $\pi(\bm\vartheta)$ for their prior.
The intercepts $\bm\nu$ and $\bm\alpha$ carry priors like any other element of $\bm\vartheta$, and the mean structure is estimated throughout this paper.
MCMC implementations do the same and supply intercepts whether or not the user asks for them \parencite{merkle2018blavaan}, which keeps the comparisons of Section~\ref{sec:sim} like for like.
This is a convention, not a restriction.
The means may instead be marginalised under a flat prior, in which case each unit is scored by its conditional density given the rest of the sample.
That kernel is Gaussian as well, and everything derived below applies to it unchanged.

Integrating out the latent state $\bm\eta_s$ yields $\mathbf{y}_s \mid \bm\vartheta \sim \mathcal{N}_p\big(\bm\mu(\bm\vartheta)\, , \bm\Sigma(\bm\vartheta)\big)$, the Gaussian reduced form, with
\begin{equation}
\bm\mu(\bm\vartheta) = \bm\nu + \bm\Lambda (\mathbf{I} - \mathbf{B})^{-1} \bm\alpha,
\qquad
\bm\Sigma(\bm\vartheta) = \bm\Lambda (\mathbf{I} - \mathbf{B})^{-1} \bm\Psi (\mathbf{I} - \mathbf{B})^{-\top} \bm\Lambda^{\top} + \bm\Theta.
\label{eq:reduced-single}
\end{equation}
This reduced form makes explicit that the structural matrix $\mathbf{B}$ acts on the data only through the convolution $(\mathbf{I} - \mathbf{B})^{-1} \bm\Psi (\mathbf{I} - \mathbf{B})^{-\top}$, the same operation that, applied once at each level, underlies the cluster reduced form.
In the multilevel case, subjects $s = 1, \dots, n_j$ are nested in clusters $j = 1, \dots, J$, and we adopt the standard within-between decomposition \autocite{muthen1994multilevel}
\begin{equation*}
\mathbf{y}_{sj} = \bm\nu + \mathbf{y}_j^{\mathrm{b}} + \mathbf{y}_{sj}^{\mathrm{w}},
\end{equation*}
in which the within-cluster component $\mathbf{y}_{sj}^{\mathrm{w}}$ and the between-cluster component $\mathbf{y}_j^{\mathrm{b}}$ are each governed by a measurement-and-structural model of the form \eqref{eq:meas-single}--\eqref{eq:struc-single}.
At the within level,
\begin{equation*}
\mathbf{y}_{sj}^{\mathrm{w}} = \bm\Lambda^{\mathrm{w}} \bm\eta_{sj}^{\mathrm{w}} + \bm\epsilon_{sj}^{\mathrm{w}}, \quad \bm\epsilon_{sj}^{\mathrm{w}} \sim \mathcal{N}_p(0, \bm\Theta^{\mathrm{w}}),
\qquad
\bm\eta_{sj}^{\mathrm{w}} = \mathbf{B}^{\mathrm{w}} \bm\eta_{sj}^{\mathrm{w}} + \bm\zeta_{sj}^{\mathrm{w}}, \quad \bm\zeta_{sj}^{\mathrm{w}} \sim \mathcal{N}_{q_{\mathrm{w}}}(0, \bm\Psi^{\mathrm{w}}),
\end{equation*}
with within-level loadings $\bm\Lambda^{\mathrm{w}} \in \R^{p \times q_{\mathrm{w}}}$, within-level latent state $\bm\eta_{sj}^{\mathrm{w}} \in \R^{q_{\mathrm{w}}}$, structural matrix $\mathbf{B}^{\mathrm{w}} \in \R^{q_{\mathrm{w}} \times q_{\mathrm{w}}}$, and disturbance covariance $\bm\Psi^{\mathrm{w}}$.
The within-level latent state carries no intercept, since within-cluster latent variables are deviations from their cluster means and so have mean zero by construction.
At the between level,
\begin{equation}
\mathbf{y}_j^{\mathrm{b}} = \bm\Lambda^{\mathrm{b}} \bm\eta_j^{\mathrm{b}} + \bm\epsilon_j^{\mathrm{b}}, \quad \bm\epsilon_j^{\mathrm{b}} \sim \mathcal{N}_p(0, \bm\Theta^{\mathrm{b}}),
\qquad
\bm\eta_j^{\mathrm{b}} = \bm\alpha^{\mathrm{b}} + \mathbf{B}^{\mathrm{b}} \bm\eta_j^{\mathrm{b}} + \bm\zeta_j^{\mathrm{b}}, \quad \bm\zeta_j^{\mathrm{b}} \sim \mathcal{N}_{q_{\mathrm{b}}}(0, \bm\Psi^{\mathrm{b}}),
\label{eq:between-model}
\end{equation}
with between-level loadings $\bm\Lambda^{\mathrm{b}} \in \R^{p \times q_{\mathrm{b}}}$, between-level latent state $\bm\eta_j^{\mathrm{b}} \in \R^{q_{\mathrm{b}}}$, structural intercept $\bm\alpha^{\mathrm{b}} \in \R^{q_{\mathrm{b}}}$, structural matrix $\mathbf{B}^{\mathrm{b}} \in \R^{q_{\mathrm{b}} \times q_{\mathrm{b}}}$, and disturbance covariance $\bm\Psi^{\mathrm{b}}$.
As in the single-level case, the matrices $\mathbf{B}^{\mathrm{w}}$ and $\mathbf{B}^{\mathrm{b}}$ collect the structural regression coefficients.
The $\bm\vartheta$ parameters are now expanded to include the within- and the between-cluster parameters $\{\bm\nu, \bm\Lambda^{\mathrm{w}}, \mathbf{B}^{\mathrm{w}}, \bm\Psi^{\mathrm{w}}, \bm\Theta^{\mathrm{w}}, \bm\Lambda^{\mathrm{b}}, \bm\alpha^{\mathrm{b}}, \mathbf{B}^{\mathrm{b}}, \bm\Psi^{\mathrm{b}}, \bm\Theta^{\mathrm{b}}\}$.
In a similar reduced-form fashion, the respective corresponding model-implied covariances can be formed as
\begin{equation}
\begin{aligned}
\bm\Sigma_{\mathrm{w}}(\bm\vartheta) &= \bm\Lambda^{\mathrm{w}} (\mathbf{I} - \mathbf{B}^{\mathrm{w}})^{-1} \bm\Psi^{\mathrm{w}} (\mathbf{I} - \mathbf{B}^{\mathrm{w}})^{-\top} (\bm\Lambda^{\mathrm{w}})^\top + \bm\Theta^{\mathrm{w}}, \\
\bm\Sigma_{\mathrm{b}}(\bm\vartheta) &= \bm\Lambda^{\mathrm{b}} (\mathbf{I} - \mathbf{B}^{\mathrm{b}})^{-1} \bm\Psi^{\mathrm{b}} (\mathbf{I} - \mathbf{B}^{\mathrm{b}})^{-\top} (\bm\Lambda^{\mathrm{b}})^\top + \bm\Theta^{\mathrm{b}},
\end{aligned}
\label{eq:Sigma-wb}
\end{equation}
both $p \times p$, while the model-implied mean as $\bm\mu(\bm\vartheta) = \bm\nu + \bm\Lambda^{\mathrm{b}} (\mathbf{I} - \mathbf{B}^{\mathrm{b}})^{-1} \bm\alpha^{\mathrm{b}}$, to which the within level contributes nothing.
Stacking the observations within a cluster as $\mathbf{y}_j = (\mathbf{y}_{1j}^\top, \dots, \mathbf{y}_{n_j j}^\top)^\top \in \R^{n_j p}$, the \emph{cluster reduced form} derived by \textcite{rosseel2021evaluating} is
\begin{equation}
\mathbf{y}_j \mid \bm\vartheta \sim \mathcal{N}_{n_j p}\bigl(\one_{n_j} \otimes \bm\mu(\bm\vartheta),\, \mathbf{V}_j(\bm\vartheta)\bigr),
\qquad
\mathbf{V}_j(\bm\vartheta) = \mathbf{I}_{n_j} \otimes \bm\Sigma_{\mathrm{w}} + \mathbf{J}_{n_j} \otimes \bm\Sigma_{\mathrm{b}},
\label{eq:cluster-reduced}
\end{equation}
with $\mathbf{V}_j(\bm\vartheta) \in \R^{n_j p \times n_j p}$ and $\mathbf{J}_{n_j} = \one_{n_j} \one_{n_j}^\top$.
The cluster covariance is compound-symmetric in block form, with $\bm\Sigma_{\mathrm{w}}$ on the diagonal of each unit's block and the shared between-cluster effect $\bm\eta_j^{\mathrm{b}}$ adding $\bm\Sigma_{\mathrm{b}}$ to every cross-unit block within the cluster.
The structural matrices $\mathbf{B}^{\mathrm{w}}$ and $\mathbf{B}^{\mathrm{b}}$ are confined entirely to the maps $\bm\vartheta \mapsto \bm\Sigma_{\mathrm{w}}$ and $\bm\vartheta \mapsto \bm\Sigma_{\mathrm{b}}$ in \eqref{eq:Sigma-wb} and to $\bm\vartheta \mapsto \bm\mu$.
The clusters are, moreover, conditionally independent given $\bm\vartheta$,
\begin{equation}
p(\mathbf{y} \mid \bm\vartheta) = \prod_{j=1}^J p(\mathbf{y}_j \mid \bm\vartheta),
\label{eq:cluster-indep}
\end{equation}
which will underwrite the cluster-deleted posterior identity of Section~\ref{sec:posterior}.

Evaluated naively, the $n_j p$-dimensional Gaussian \eqref{eq:cluster-reduced} costs $O\bigl((n_j p)^3\bigr)$ per cluster; but as \textcite{rosseel2021evaluating} notes, the compound-symmetric structure of $\mathbf{V}_j$ removes the dependence on cluster size, since $\mathbf{V}_j$ decomposes over the within- and between-cluster subspaces spanned by the orthogonal projectors $\mathbf{I}_{n_j} - n_j^{-1} \mathbf{J}_{n_j}$ and $n_j^{-1} \mathbf{J}_{n_j}$. 
Every log-determinant, inverse, and quadratic form then reduces to operations on the two $p \times p$ matrices $\bm\Sigma_{\mathrm{w}}$ and $\bm\Sigma_{\mathrm{w}} + n_j \bm\Sigma_{\mathrm{b}}$.
With $\bar{\mathbf{y}}_j = n_j^{-1} \sum_{i=1}^{n_j} \mathbf{y}_{ij}$ the cluster mean and $\mathbf{S}_j = \sum_{i=1}^{n_j} (\mathbf{y}_{ij} - \bar{\mathbf{y}}_j)(\mathbf{y}_{ij} - \bar{\mathbf{y}}_j)^\top$ the within-cluster scatter matrix, the Gaussian log-likelihood splits along the same two subspaces,
\begin{equation}
\log p(\mathbf{y}_j \mid \bm\vartheta) = -\tfrac{1}{2}\Bigl[\, n_j p \log(2\pi) + \log \det \mathbf{V}_j + Q_j^{\mathrm{w}}(\bm\vartheta) + Q_j^{\mathrm{b}}(\bm\vartheta) \,\Bigr],
\label{eq:cluster-loglik}
\end{equation}
with $\log \det \mathbf{V}_j = (n_j - 1) \log \det \bm\Sigma_{\mathrm{w}} + \log \det\bigl(\bm\Sigma_{\mathrm{w}} + n_j \bm\Sigma_{\mathrm{b}}\bigr)$ and within- and between-cluster quadratic forms
\begin{equation*}
Q_j^{\mathrm{w}}(\bm\vartheta) = \tr\bigl(\bm\Sigma_{\mathrm{w}}^{-1} \mathbf{S}_j\bigr),
\qquad
Q_j^{\mathrm{b}}(\bm\vartheta) = n_j (\bar{\mathbf{y}}_j - \bm\mu)^\top \bigl(\bm\Sigma_{\mathrm{w}} + n_j \bm\Sigma_{\mathrm{b}}\bigr)^{-1} (\bar{\mathbf{y}}_j - \bm\mu).
\end{equation*}
The data enter only through the sufficient statistics $(\bar{\mathbf{y}}_j, \mathbf{S}_j, n_j)$, formed once before inference, and each subsequent evaluation of the cluster log-likelihood $\log p(\mathbf{y}_j \mid \bm\vartheta)$ then costs $O(p^3)$ for two matrix factorisations and $O(n_j p^2)$ for the quadratic forms, with the cubic dependence on cluster size eliminated.
The decomposition presumes that every coordinate varies within clusters, so that $\bm\Sigma_{\mathrm{w}}$ is positive definite.
A cluster-constant coordinate occupies a null block of $\bm\Sigma_{\mathrm{w}}$ and calls instead for the general two-level Gaussian likelihood of \textcite[Eqs. 21--24]{rosseel2021evaluating}, of which \eqref{eq:cluster-loglik} is the special case with no between-only coordinates.

Exogenous covariates are accommodated by modelling them jointly with the indicators, in the manner of the all-$\mathbf{y}$ parameterisation \parencite{bollen1989structural}.
That is, a covariate $\mathbf{x}_{sj}$ is appended to the response vector, so its mean, variance, and covariances become free parameters, and the cluster reduced form \eqref{eq:cluster-reduced} and the conditional independence \eqref{eq:cluster-indep} continue to hold with $p$ enlarged to include the covariates.
Covariates that vary within clusters enter both $\bm\Sigma_{\mathrm{w}}$ and $\bm\Sigma_{\mathrm{b}}$ and require nothing new, whereas cluster-level covariates are constant within a cluster, contribute no within-cluster variation, and enter only the between covariance $\bm\Sigma_{\mathrm{b}}$, a case handled by the general two-level Gaussian likelihood.
The alternative convention of conditioning on the covariates induces a different predictive estimand, set out in Appendix~\ref{app:covariates} together with the exact relation between the two.

For multilevel SEM, observation-level LOO can be cast as either of two distinct prediction tasks, differing only in the distribution against which the cluster random effect $\bm\eta_j^{\mathrm{b}}$ is integrated.
The \emph{conditional} (\loso{}) target predicts a new subject in an observed cluster: $\bm\eta_j^{\mathrm{b}}$ is integrated against its posterior $p(\bm\eta_j^{\mathrm{b}} \mid \mathbf{y}_{j,-s}, \bm\vartheta)$, so the held-out subject's own clustermates inform the cluster effect and the prediction borrows their strength,
\begin{equation}
p(\mathbf{y}_{sj} \mid \mathbf{y}_{j,-s}, \mathbf{y}_{-j}) = \iint p(\mathbf{y}_{sj} \mid \bm\eta_j^{\mathrm{b}}, \bm\vartheta)\, p(\bm\eta_j^{\mathrm{b}} \mid \mathbf{y}_{j,-s}, \bm\vartheta)\, d\bm\eta_j^{\mathrm{b}}\, \pi(\bm\vartheta \mid \mathbf{y}_{-(sj)})\, d\bm\vartheta.
\label{eq:loso-target}
\end{equation}
The \emph{marginal} (\loco{}) target predicts an entirely new cluster: there are no clustermates to estimate $\bm\eta_j^{\mathrm{b}}$ from, so it is integrated against its prior, the marginal latent distribution $\mathcal{N}_{q_{\mathrm{b}}}\bigl((\mathbf{I} - \mathbf{B}^{\mathrm{b}})^{-1}\bm\alpha^{\mathrm{b}},\, (\mathbf{I} - \mathbf{B}^{\mathrm{b}})^{-1} \bm\Psi^{\mathrm{b}} (\mathbf{I} - \mathbf{B}^{\mathrm{b}})^{-\top}\bigr)$ implied by the between-level structural model \eqref{eq:between-model},
\begin{equation}
p(\mathbf{y}_j \mid \mathbf{y}_{-j}) = \iint p(\mathbf{y}_j \mid \bm\eta_j^{\mathrm{b}}, \bm\vartheta)\, p(\bm\eta_j^{\mathrm{b}} \mid \bm\vartheta)\, d\bm\eta_j^{\mathrm{b}}\, \pi(\bm\vartheta \mid \mathbf{y}_{-j})\, d\bm\vartheta = \int p(\mathbf{y}_j \mid \bm\vartheta)\, \pi(\bm\vartheta \mid \mathbf{y}_{-j})\, d\bm\vartheta.
\label{eq:loco-target}
\end{equation}
The second equality uses the cluster reduced form \eqref{eq:cluster-reduced}, in which this latent integral has already been carried out.
The choice between the two is not a matter of taste but of which generalisation question is being asked.
Because subjects within a cluster are correlated through the shared effect $\bm\eta_j^{\mathrm{b}}$, leaving out a single subject leaves that subject's cluster effect still estimable from the rest of the cluster.
\loso{} therefore rewards a model for fitting cluster-specific idiosyncrasy, which does not transfer to a new population.
Leaving out an entire cluster removes every observation that identifies $\bm\eta_j^{\mathrm{b}}$, so \loco{} rewards only population-level structure that generalises across clusters.
We therefore take \eqref{eq:loco-target} as our target and summarise it by the summed log-predictive score
\begin{equation*}
\elpd_{\loco} = \sum_{j=1}^J \log p(\mathbf{y}_j \mid \mathbf{y}_{-j}),
\end{equation*}
with larger values indicating better generalisation to a new cluster.

\section{Leave-One-Cluster-Out From a Single Model Fit}
\label{sec:posterior}

The \loco{} target \eqref{eq:loco-target} is an integral against the cluster-deleted posterior $\pi(\bm\vartheta \mid \mathbf{y}_{-j})$, and computing that posterior directly by refitting the model for each of the $J$ clusters is precisely the computational bottleneck this framework eliminates.
The next two subsections remove it, reducing the target to one expectation against the full-data posterior and then evaluating that expectation in closed form.
The two after them repeat the pattern one level down, for a candidate submodel scored by conditioning the same Laplace summary rather than by its own fit.

\subsection{Exact Cluster-Deletion Identity}
\label{sec:harmonic}

As seen from \eqref{eq:cluster-indep}, cluster $j$ enters the full-data likelihood through a single factor, so deleting it reweights the full posterior rather than requiring a fresh fit.
The following well-known proposition bypasses these refits entirely, collapsing \eqref{eq:loco-target} into a single expectation against the \emph{full-data} posterior.

\begin{proposition}[\loco{} as a posterior harmonic mean]
\label{prop:harmonic-mean}
If the full-data likelihood factorises at cluster $j$ as $p(\mathbf{y} \mid \bm\vartheta) = p(\mathbf{y}_j \mid \bm\vartheta)\, p(\mathbf{y}_{-j} \mid \bm\vartheta)$, as conditional independence of the clusters \eqref{eq:cluster-indep} guarantees, then for every cluster $j$ the \loco{} predictive density \eqref{eq:loco-target} satisfies the exact identity
\begin{equation}
p(\mathbf{y}_j \mid \mathbf{y}_{-j}) = \Bigl(\, \E_{\bm\vartheta \mid \mathbf{y}} \bigl[\, 1 / p(\mathbf{y}_j \mid \bm\vartheta) \,\bigr] \,\Bigr)^{-1},
\label{eq:harmonic-mean}
\end{equation}
where the expectation is taken against the full-data posterior $\pi(\bm\vartheta \mid \mathbf{y})$.
\end{proposition}

\begin{proof}
By the assumed factorisation, Bayes' theorem immediately yields the cluster-deleted posterior as a direct reweighting of the full posterior,
\begin{equation}
\pi(\bm\vartheta \mid \mathbf{y}_{-j}) \propto  \frac{\pi(\bm\vartheta \mid \mathbf{y})}{p(\mathbf{y}_j \mid \bm\vartheta)},
\label{eq:cluster-deleted-id}
\end{equation}
with normalising constant $Z_j := \E_{\bm\vartheta \mid \mathbf{y}} \left[ \frac{1}{p(\mathbf{y}_j \mid \bm\vartheta)} \right]$. 
Normalising and substituting \eqref{eq:cluster-deleted-id} into the predictive target \eqref{eq:loco-target} yields
\begin{equation*}
p(\mathbf{y}_j \mid \mathbf{y}_{-j})
= \int \cancel{p(\mathbf{y}_j \mid \bm\vartheta)} \frac{1}{Z_j} \frac{\pi(\bm\vartheta \mid \mathbf{y})}{\cancel{p(\mathbf{y}_j \mid \bm\vartheta)}} \, d\bm\vartheta
= \frac{1}{Z_j} \int \pi(\bm\vartheta \mid \mathbf{y})\, d\bm\vartheta
= Z_j^{-1},
\end{equation*}
the last step using the fact that $\pi(\bm\vartheta \mid \mathbf{y})$ integrates to one. This is \eqref{eq:harmonic-mean}.
\end{proof}

The right-hand quantity in \eqref{eq:harmonic-mean} is the multivariate cluster-level CPO, and Proposition~\ref{prop:harmonic-mean} delivers this target in a tractable form.
Every cluster-deleted refit is gone; the right-hand side is a functional of the single posterior $\pi(\bm\vartheta \mid \mathbf{y})$, so one encompassing fit scores all $J$ clusters at once.
The reduction is also exact and analytic.
For comparison, the leave-group-out construction of \textcite{liu2025leavegroupout} for general latent Gaussian models targets group likelihoods that are typically non-Gaussian, and so retains an inner numerical integration over the latent state.
Here, the Gaussian cluster reduced form \eqref{eq:cluster-reduced} makes that inner integral closed-form, and Proposition~\ref{prop:harmonic-mean} collapses the outer integral over $\bm\vartheta$ as well.
The \loco{} score for the multilevel Gaussian SEM is therefore both refit-free and quadrature-free.

What remains is to evaluate the posterior expectation in \eqref{eq:harmonic-mean}.
Even if posterior samples are readily available, either as a by-product of MCMC or from the Gaussian copula draws that \texttt{INLAvaan} supplies \autocite{jamil2026implementation}, a direct Monte Carlo average over them must be avoided.
This is because such an average is the cluster-level counterpart of the harmonic-mean estimator of the marginal likelihood, and inherits its variance pathology.
The integrand $1/p(\mathbf{y}_j \mid \bm\vartheta)$ has heavy tails wherever the cluster likelihood is small over a region of non-negligible posterior mass, and the resulting Monte Carlo estimator can have infinite variance \autocite{wolpert2012alpha}.
The next section instead evaluates \eqref{eq:harmonic-mean} deterministically, by a Laplace approximation that has zero variance by construction and so sidesteps the pathology entirely, at a cost far below that of stabilising a Monte Carlo estimate.

\subsection{A Deterministic Laplace Approximation of the Cluster CPO}
\label{sec:laplace}

We turn to the expectation, approximating it by a fully exponential Laplace approximation carried out with the cluster log-likelihood \emph{inside} the exponent.
Let $L(\bm\vartheta) = \log \pi(\bm\vartheta \mid \mathbf{y})$ and let $\mathcal{N}(\bm\vartheta^\ast, \bm\Omega)$ be its Gaussian summary, in which $\bm\vartheta^\ast$ is the centre at which the expansion is taken and $\bm\Omega^{-1} := -\nabla^2 L$ the curvature at the posterior mode.
For each cluster define the score and Hessian at that centre, respectively $\mathbf{s}_j := \nabla \ell_j(\bm\vartheta^\ast)$ and $\mathbf{H}_j := \nabla^2 \ell_j(\bm\vartheta^\ast)$, where $\ell_j(\bm\vartheta) := \log p(\mathbf{y}_j \mid \bm\vartheta)$ is the cluster log-likelihood.
Set $\bm\delta := \bm\vartheta - \bm\vartheta^\ast$ and $d := \dim \bm\vartheta$, and write
\begin{equation*}
I_j := \E_{\bm\vartheta \mid \mathbf{y}} \bigl[\, e^{-\ell_j(\bm\vartheta)} \,\bigr], \qquad \hat{p}(\mathbf{y}_j \mid \mathbf{y}_{-j}) := 1 / I_j,
\end{equation*}
which by Proposition~\ref{prop:harmonic-mean} is exact up to the approximation of the single expectation $I_j$.
The following theorem evaluates $I_j$ exactly under a Taylor truncation of $\ell_j$, taken to first order and to second.

\begin{theorem}[Laplace approximation of the cluster CPO]
\label{prop:laplace-1st}
\label{prop:laplace-2nd}
Assume that the posterior is replaced by its Gaussian summary, $\pi(\bm\vartheta \mid \mathbf{y}) = \mathcal{N}(\bm\vartheta^\ast, \bm\Omega)$ with $\bm\Omega \succ 0$.
Write $\hat{I}_j^{(k)}$ for $I_j$ with $\ell_j$ replaced by its $k$th-order Taylor expansion about $\bm\vartheta^\ast$, and set $\hat{p}^{(k)}(\mathbf{y}_j \mid \mathbf{y}_{-j}) := 1 / \hat{I}_j^{(k)}$. Then:
\begin{enumerate}
\item[(i)] \textbf{\emph{First order.}} $\hat{I}_j^{(1)}$ is finite for every finite $\mathbf{s}_j$, and
\begin{equation}
\log \hat{p}^{(1)}(\mathbf{y}_j \mid \mathbf{y}_{-j}) =
\ell_j(\bm\vartheta^\ast) - \tfrac{1}{2}\, \mathbf{s}_j^\top \bm\Omega \mathbf{s}_j.
\label{eq:laplace-1st}
\end{equation}
\item[(ii)] \textbf{\emph{Second order.}} With $\mathbf{A}_j := \bm\Omega^{-1} + \mathbf{H}_j$, $\hat{I}_j^{(2)}$ is finite if and only if
\begin{equation}
\mathbf{A}_j \succ 0
\quad \Longleftrightarrow \quad
\lambda_{\max}\bigl(-\bm\Omega^{1/2} \mathbf{H}_j \bm\Omega^{1/2}\bigr) < 1,
\label{eq:existence-2nd}
\end{equation}
and in that case
\begin{equation}
\log \hat{p}^{(2)}(\mathbf{y}_j \mid \mathbf{y}_{-j}) =
\ell_j(\bm\vartheta^\ast) - \tfrac{1}{2}\, \mathbf{s}_j^\top \mathbf{A}_j^{-1} \mathbf{s}_j + \tfrac{1}{2}\, \log \bigl| \mathbf{I} + \bm\Omega \mathbf{H}_j \bigr|.
\label{eq:laplace-2nd}
\end{equation}
\end{enumerate}
Part (i) is the case $\mathbf{H}_j = 0$ of (ii).
\end{theorem}

The proof is given in Appendix~\ref{app:laplace-proof}.
Under the Gaussian summary, then, the fully exponential Laplace approximation of the cluster CPO is nothing more than a Taylor truncation of $\ell_j$, and we refer to $\log \hat{p}^{(1)}$ and $\log \hat{p}^{(2)}$ as the first- and second-order \emph{Taylor elpd scores} throughout.
The second-order form differs from \eqref{eq:laplace-1st} through a single precision update $\bm\Omega^{-1}\mapsto \mathbf{A}_j = \bm\Omega^{-1}+\mathbf{H}_j$ and adds a determinant term for the resulting change in normalising volume.
These adjustments have a direct deletion interpretation.
Removing cluster $j$ subtracts $\ell_j$ from the full log posterior, changing its negative Hessian at $\bm\vartheta^\ast$ to the full-data precision less the observed information $-\mathbf{H}_j$ contributed by cluster $j$, so $\mathbf A_j$ is interpreted as the approximate leave-$j$-out precision.
Accordingly, $\bm\vartheta^\ast-\mathbf{A}_j^{-1}\mathbf{s}_j$ is one analytic Newton step towards the cluster-deleted mode, i.e., the tilted mode of \textcite{tierney1986accurate}, with $\mathbf{A}_j$ the tilted precision, reached in a single step because of the Gaussian approximation to the posterior.
Whereas \textcite{tierney1986accurate} remaximise the cluster-deleted log posterior $L - \ell_j$ and evaluate its Hessian at the new mode, here the tilted precision $\mathbf{A}_j$ is fixed at $\bm\vartheta^\ast$, so the second-order score is their fully exponential approximation applied to the quadratic surrogates of $L$ and $L - \ell_j$.

Condition~\eqref{eq:existence-2nd}, $\mathbf{A}_j\succ0$, is therefore the existence condition for this local Gaussian representation, as
it makes $\mathbf{A}_j$ a valid precision matrix and ensures convergence of the Gaussian integral defining the second-order score.
Let $k_{ji}$ be the eigenvalues of $-\bm\Omega^{1/2}\mathbf{H}_j\bm\Omega^{1/2}$.
$\mathbf{A}_j\succ0$ holds exactly when $k_j:=\max_i k_{ji}<1$  for any $\mathbf{H}_j$, since $\mathbf{A}_j\succ0$ is equivalent to $\mathbf{I}+\bm\Omega^{1/2}\mathbf{H}_j\bm\Omega^{1/2}\succ0$, and the determinant adjustment becomes $\tfrac12\sum_i\log(1-k_{ji})$, an additional volume penalty whenever the $k_{ji}$ are non-negative, as they are when $\mathbf{H}_j\preceq0$. 
The condition $k_j < 1$ also governs the harmonic-mean pathology described earlier.
Since $I_j$ is exactly the posterior mean of the deleted-to-full density ratio $r_j(\bm\vartheta) \propto p(\mathbf{y}_j \mid \bm\vartheta)^{-1}$ in \eqref{eq:cluster-deleted-id}, the existence condition just derived is also the condition for that ratio to have finite mean.

Both statements in Theorem~\ref{prop:laplace-2nd} are exact, so the error in $\log \hat p^{(k)}$ is entirely the error of the Taylor truncation of $\ell_j$.
With $J \to \infty$ at $n_j$ and $p$ fixed, that error is $O(J^{-1})$ per cluster at first order (the discarded second-order expansion term) and $O(J^{-2})$ at second.
Thus, the two orders differ per cluster by $\tfrac12 \operatorname{tr}(-\bm\Omega \mathbf{H}_j) + O(J^{-2})$, the leading term arising from the determinant adjustment and equalling minus the posterior mean of the quadratic Taylor term the first-order truncation discards, the tilt of the quadratic form entering only at $O(J^{-2})$.
Now $\elpd_{\loco}$ sums $J$ of these, and provided the bounds hold uniformly in $j$, the summed second-order truncation error is $O(J^{-1})$ and vanishes as clusters accumulate, whereas the summed first-order error is $O(1)$ and does not.
That non-vanishing limit is the sum of the per-cluster gaps, since $\sum_j \mathbf{H}_j$ is the full-data Hessian $-\mathcal{I}$ and $\bm\Omega^{-1} = \mathcal{I} + \mathbf{P}$, with $\mathbf{P} := -\nabla^2 \log \pi(\bm\vartheta^\ast)$ the curvature the prior contributes there,
\begin{equation}
\sum_j \Bigl( \log \hat{p}^{(1)}(\mathbf{y}_j \mid \mathbf{y}_{-j}) - \log \hat{p}^{(2)}(\mathbf{y}_j \mid \mathbf{y}_{-j}) \Bigr) \;=\; \tfrac{1}{2} \operatorname{tr}\bigl( \bm\Omega \mathcal{I} \bigr) + O(J^{-1}),
\label{eq:peff-gap}
\end{equation}
the trace being the effective number of parameters $p_{\mathrm{eff}} := \operatorname{tr}(\bm\Omega \mathcal{I})$ familiar from the deviance information criterion \parencite[DIC;][]{spiegelhalter2002bayesian} and the widely applicable information criterion \parencite[WAIC;][]{watanabe2010asymptotic}.
It is the dimension $d$ less a prior shrinkage, in the manner of a ridge trace, since $p_{\mathrm{eff}} = d - \operatorname{tr}(\bm\Omega \mathbf{P}) \le d$ with equality only under a flat prior. 
That shrinkage is itself $O(J^{-1})$, as is the remainder in \eqref{eq:peff-gap}, so the summed gap approaches $d/2$ while differing from it at any finite $J$.
The first-order Taylor elpd therefore overstates $\elpd_{\loco}$ by a complexity penalty $\tfrac12 p_{\mathrm{eff}}$, and so comparing candidates of different dimension by the first-order score may be misleading.
The same per-cluster kernels also deliver the marginal WAIC, which targets this same leave-one-cluster-out quantity \parencite{merkle2019bayesian} and is ordinarily computed from posterior draws.
Under a Gaussian density it needs none, and at first order it reproduces \eqref{eq:laplace-1st} exactly, its penalty being the $p_{\mathrm{loo}}$ that the \texttt{loo} package estimates by PSIS \autocite{vehtari2017practical}.
Appendix~\ref{app:waic} gives the derivation.

The centre $\bm\vartheta^\ast$ need not be the posterior mode.
Theorem~\ref{prop:laplace-1st} requires only that $\ell_j$ be expanded at the centre of the Gaussian summary, so that $\bm\delta$ has mean zero under it.
The stationarity condition $\nabla L(\bm\vartheta^\ast) = 0$ is nowhere used, and any centre is admissible. 
Since $I_j$ is an expectation under the posterior, a centre near the posterior mean is preferable to the mode whenever the posterior is skewed. 
\texttt{INLAvaan} supplies a variationally mean-corrected centre \parencite{jamil2026approximate}, and it is that point at which we therefore expand, though $\bm\Omega$ retains the curvature at the mode.
Every remaining ingredient comes from that same fit.
The value $\ell_j(\bm\vartheta^\ast)$ is closed-form from \eqref{eq:cluster-loglik}, at $O(p^3)$ per cluster and independent of cluster size.
The score $\mathbf{s}_j$ follows by analytic differentiation of $\ell_j$ through the matrix calculus of $\bm\Sigma_{\mathrm{w}}(\bm\vartheta)$, $\bm\Sigma_{\mathrm{b}}(\bm\vartheta)$ and $\bm\mu(\bm\vartheta)$, and the Hessian $\mathbf{H}_j$, needed only for the second-order form, by finite differencing that analytic gradient.
The total cost is one \texttt{INLAvaan} fit followed by a single post-processing pass, with no refits and no posterior sampling.

\subsection{Submodels as Linear Constraints and the Compatible-Prior Identity}
\label{sec:restrictions}

The Laplace summary $\mathcal{N}(\bm\vartheta^\ast, \bm\Omega)$ of Section~\ref{sec:laplace} supports more than \loco{} scoring of the encompassing model. 
In the spirit of projection predictive inference \citep[e.g.,][]{piironen2020projective}, we bypass the prohibitive computational cost of combinatoric refitting by treating our fully parameterised encompassing model as an information-rich reference profile. 
A candidate structural restriction is then handled as a linear constraint on $\bm\vartheta$.
Under a compatibility convention on the submodel priors, the restricted Laplace summary follows from closed-form Gaussian conditioning, and the \loco{} score of the restricted model is a one-pass post-processing step. 
Evaluating a candidate submodel thus costs one conditioning update and one scoring pass, against one full \texttt{INLAvaan} fit for the encompassing model, so a candidate space of realistic size can be enumerated rather than searched.

To evaluate structural alternatives efficiently, we operate within a framework where the candidate space consists of candidate submodels nested within a single, unrestricted reference model. 
Let this encompassing model be parameterised by the full vector $\bm\vartheta$. Any specific candidate structural submodel $\mathcal{M}_m$ is then identified with the linear constraint $\mathbf{C}_m\bm\vartheta = 0$, in which $\mathbf{C}_m$ is a selector matrix whose rows pick out the path entries zeroed under $\mathcal{M}_m$.
Single-path additions or removals during structure selection thus correspond to the sequential removal or addition of a single row in $\mathbf{C}_m$, respectively.
The representation presumes that zeroing the candidate paths leaves the remaining model identified.
Restrictions that change identification status (e.g., deleting a loading that leaves a latent factor underdetermined) lie outside its scope and require a separate \texttt{INLAvaan} fit.
The constraint acts on $\bm\vartheta$ alone, and so composes with any deletion index, including subject-level leave-out.

We adopt throughout the \emph{compatible-prior convention} of \textcite{consonni2008compatibility}, taking each submodel's prior to be the encompassing prior conditioned on its own defining constraint,
\begin{equation}
\pi_m(\bm\vartheta) = \pi(\bm\vartheta \mid \mathbf{C}_m \bm\vartheta = 0).
\label{eq:compat}
\end{equation}
This is the same condition under which the Savage--Dickey density ratio is valid \parencite{dickey1971weighted,verdinelli1995computing} and that underlies the encompassing-prior approach to constrained model selection \parencite{klugkist2007bayes}.

Compatibility is automatic in the cases that arise here.
\texttt{INLAvaan} places independent priors on the parameters by default, so when $\pi(\bm\vartheta) = \prod_k \pi_k(\vartheta_k)$ and $\mathbf{C}_m$ is a selector matrix, conditioning on $\mathbf{C}_m\bm\vartheta = 0$ leaves the priors on the retained parameters untouched and \eqref{eq:compat} holds with $\pi_m$ their encompassing marginal prior.
The argument requires only that the constrained coordinates be prior-independent of the retained ones, so a joint prior on a block of parameters, such as one on a latent correlation matrix \autocite{freni2025graphical,freni2026informative}, is compatible as long as no coordinate inside that block is zeroed.
Since structure selection adds and removes individual paths, only selector constraints arise and no separate submodel prior need be elicited.

\begin{lemma}
\label{lem:submodel-id}
Partition $\bm\vartheta = (\vartheta_a, \bm\vartheta_{-a})$ and consider the single-path restriction $\vartheta_a = 0$.
Let $\mathcal{M}_m$ denote the submodel with $\vartheta_a$ fixed at zero, carrying prior $\pi_m(\bm\vartheta_{-a})$.
If the compatible-prior convention $\pi_m(\bm\vartheta_{-a}) = \pi(\bm\vartheta_{-a} \mid \vartheta_a = 0)$ holds, then for every dataset $\mathbf{y}$
\begin{equation}
\pi_m(\bm\vartheta_{-a} \mid \mathbf{y}) \;=\; \pi(\bm\vartheta_{-a} \mid \mathbf{y},\, \vartheta_a = 0).
\label{eq:submodel-id}
\end{equation}
The general linear restriction $\mathbf{C}_m \bm\vartheta = 0$ reduces to this case by an invertible linear reparameterisation mapping $\mathbf{C}_m \bm\vartheta$ onto a subvector of the parameters.
\end{lemma}

The identity follows from Bayes' theorem once the joint prior is factored on the constraint slice, the mass at $\vartheta_a = 0$ absorbing into the normalising constant.
Appendix~\ref{app:compat} gives the argument in full.
The posterior equality \eqref{eq:submodel-id} is used implicitly in those Savage--Dickey derivations to compute Bayes factors, whereas here it does something different, licensing the refit-free construction.
Because the \loco{} score of Section~\ref{sec:laplace} is a  functional of the posterior alone, replacing the submodel refit by conditioning the encompassing posterior introduces no prior-level approximation.
Of course, while prior-level compatibility is exact, our operational machinery  applies the constraint $\mathbf{C}_m\bm\vartheta = 0$ directly to the Gaussian Laplace summary $\mathcal{N}(\bm\vartheta^\ast, \bm\Omega)$ of the reference model.  
This corresponds to a second-order Taylor approximation of the true constrained mode and curvature. 
The size of that discrepancy depends on how far the constraint sits from the encompassing centre.
Section~\ref{sec:sim-conditioning} quantifies the dependence and shows the error to be small near the centre, one-sided beyond it, and diagnosable in advance from the Laplace summary itself.

\subsection{Scoring a Restricted Submodel}
\label{sec:restr-laplace}

Under the Laplace approximation $\pi(\bm\vartheta \mid \mathbf{y}) \approx \mathcal{N}(\bm\vartheta^\ast, \bm\Omega)$, conditioning on $\mathbf{C}_m\bm\vartheta = 0$ is Gaussian with moments
\begin{align}
\bm\vartheta^\ast_m &= \bm\vartheta^\ast - \bm\Omega \mathbf{C}_m^\top \bigl(\mathbf{C}_m \bm\Omega \mathbf{C}_m^\top\bigr)^{-1} \mathbf{C}_m\,\bm\vartheta^\ast, \label{eq:restr-mean} \\
\bm\Omega_m &= \bm\Omega - \bm\Omega \mathbf{C}_m^\top \bigl(\mathbf{C}_m \bm\Omega \mathbf{C}_m^\top\bigr)^{-1} \mathbf{C}_m \bm\Omega. \label{eq:restr-cov}
\end{align}
The matrix inverted has dimension equal to the number of constraints, typically far smaller than $d$.
\textcite{penny2013efficient} use the same operation to evaluate Savage--Dickey Bayes factors in nested generalised linear models.

For a single-path restriction $\vartheta_a = 0$, setting $\mathbf{C}_m = \mathbf{e}_a^\top$ in \eqref{eq:restr-mean}--\eqref{eq:restr-cov} gives the rank-one update
\begin{equation}
\bm\vartheta^\ast_m = \bm\vartheta^\ast - \frac{\vartheta^\ast_a}{\Omega_{aa}}\,\bm\Omega_{\cdot a},
\qquad
\bm\Omega_m = \bm\Omega - \frac{1}{\Omega_{aa}}\,\bm\Omega_{\cdot a} (\bm\Omega_{\cdot a})^\top,
\label{eq:restr-rank1}
\end{equation}
where $\bm\Omega_{\cdot a}$ is the $a$th column of $\bm\Omega$ and $\Omega_{aa}$ its $a$th diagonal entry.
The cost is $O(d^2)$ per restriction, with no matrix inversion required once $\bm\Omega$ is in hand.
This is the cheapest case, met whenever two candidate structures differ by a single dropped path.

Scoring a submodel by \eqref{eq:laplace-2nd} requires the cluster-level score and Hessian at the restricted centre $\bm\vartheta^\ast_m$.
Since $\bm\Omega_m$ is singular on the constraint span, all operations are performed in the $d_m = d - \rank(\mathbf{C}_m)$ unconstrained subspace.
Let $\bm\vartheta^\ast_{m,\mathrm{free}}$, $\bm\Omega_{m,\mathrm{free}}$, $\mathbf{s}_{j,m,\mathrm{free}}$, and $\mathbf{H}_{j,m,\mathrm{free}}$ denote the corresponding quantities.
The score and Hessian are analytic in the model-implied covariances and are obtained once per candidate; the explicit trace formulae are collected in Appendix~\ref{app:trace}.
The submodel score is then
\begin{equation}
\log \hat{p}_m(\mathbf{y}_j \mid \mathbf{y}_{-j})
\;\approx\; \ell_{j,m}(\bm\vartheta^\ast_{m,\mathrm{free}})
- \tfrac{1}{2}\,\mathbf{s}_{j,m,\mathrm{free}}^\top \mathbf{A}_{j,m}^{-1} \mathbf{s}_{j,m,\mathrm{free}}
+ \tfrac{1}{2}\,\log\bigl|\,\mathbf{I} + \bm\Omega_{m,\mathrm{free}}\,\mathbf{H}_{j,m,\mathrm{free}}\bigr|,
\label{eq:laplace-2nd-restricted}
\end{equation}
with $\mathbf{A}_{j,m} = \bm\Omega^{-1}_{m,\mathrm{free}} + \mathbf{H}_{j,m,\mathrm{free}}$, and $\elpd_{\loco}(m) = \sum_j \log \hat{p}_m(\mathbf{y}_j \mid \mathbf{y}_{-j})$.

\section{Leave-One-Subject-Out From the Same Model Fit}
\label{sec:loso-main}
Proposition~\ref{prop:harmonic-mean} and Theorem~\ref{prop:laplace-1st} use the held-out unit only through (a) a factorisation of the full-data likelihood into the density to be scored and the density of everything else, and (b) the availability of a per-unit log-likelihood with score and Hessian at the mode.
Neither property is special to clusters.
The deletion index is accordingly a modelling choice, and it is that choice, together with what the held-out density is conditioned on, that fixes the prediction question.
Where the held-out block is independent of the rest given $\bm\vartheta$, as a cluster is by \eqref{eq:cluster-indep} and a subject is under independent sampling, the required factorisation is the marginal one already used.
Where it is not (e.g., subject inside an observed cluster), the chain rule supplies it, $p(\mathbf{y} \mid \bm\vartheta) = p(\mathbf{y}_{sj} \mid \mathbf{y}_{j,-s}, \bm\vartheta)\, p(\mathbf{y}_{-(sj)} \mid \bm\vartheta)$, and the construction applies with this conditional density instead of the marginal one.
Two subject-level scores follow, the conditional target within a cluster and the marginal target for independent data.

\subsection{The Conditional Target: A Subject in an Observed Cluster}
\label{sec:loso-conditional}

The conditional \loso{} target \eqref{eq:loso-target} scores a held-out subject against the clustermates that remain.
Its integrand, $p(\mathbf{y}_{sj} \mid \mathbf{y}_{j,-s}, \bm\vartheta)$, is available in closed form via standard Gaussian conditioning as $\mathbf{y}_{sj} \mid \mathbf{y}_{j,-s}, \bm\vartheta \sim \mathcal{N}_p(\bm\mu_{sj}, \bm\Sigma_{sj})$, with
\begin{equation*}
\begin{aligned}
\bm\mu_{sj}(\bm\vartheta) &= \bm\mu + m_j \bm\Sigma_{\mathrm{b}} (\bm\Sigma_{\mathrm{w}} + m_j \bm\Sigma_{\mathrm{b}})^{-1} (\bar{\mathbf{y}}_{j,-s} - \bm\mu), \\
\bm\Sigma_{sj}(\bm\vartheta) &= \bm\Sigma_{\mathrm{w}} + \bm\Sigma_{\mathrm{b}} - m_j \bm\Sigma_{\mathrm{b}} (\bm\Sigma_{\mathrm{w}} + m_j \bm\Sigma_{\mathrm{b}})^{-1} \bm\Sigma_{\mathrm{b}},
\end{aligned}
\end{equation*}
where $m_j = n_j - 1$ and $\bar{\mathbf{y}}_{j,-s}$ is the leave-one-out cluster mean.
The harmonic-mean identity of Proposition~\ref{prop:harmonic-mean} holds with this conditional density replacing $p(\mathbf{y}_j \mid \bm\vartheta)$, and the Taylor scores in Equations~\eqref{eq:laplace-1st}--\eqref{eq:laplace-2nd} apply verbatim with $\ell_{sj}(\bm\vartheta) := \log p(\mathbf{y}_{sj} \mid \mathbf{y}_{j,-s}, \bm\vartheta)$ and its derivatives replacing $\ell_j$, $\mathbf{s}_j$, and $\mathbf{H}_j$.
The total conditional score $\elpd_{\loso}$ is thus computable from the same encompassing model, at a cost proportional to $N = \sum_j n_j$.
Note that while the gap identity \eqref{eq:peff-gap} holds for any deletion unit, its limit $d$ requires disjoint ones (a requirement met by clusters, but not conditionally scored subjects).

Where the goal is interpolation within an observed cluster, i.e., predicting a further
student in a school already in the sample, the above conditional target is the right
one.
For clustered data that is seldom the goal, however.
\textcite{merkle2019bayesian} make the choice turn on whether a model is meant to predict for the clusters in the data or for new ones, and recommend the marginal, cluster-level criterion whenever it is the latter.
The reason is the one given in Section~\ref{sec:setting}, that the held-out subject's clustermates still identify the cluster effect, so the conditional criterion rewards cluster-specific idiosyncrasy that does not transfer.
Regardless of which target a study needs, \texttt{INLAvaan} produces both from the same fit.
The \texttt{type} argument of \texttt{loo()} selects between them,
\texttt{"loco"} for the cluster-level score and \texttt{"loso"} for the
conditional subject-level one.
On two-level fits the cluster-level target is the default.

\subsection{The Marginal Target: Independent Data}
\label{sec:loso-marginal}

The marginal subject-level target is the appropriate one for independent data, where a held-out subject has no clustermates to borrow from.
For independent (single-level) data the held-out unit is a single subject $s$, the observations are independent given $\bm\vartheta$, and the cluster reduced form \eqref{eq:cluster-reduced} collapses to the single-level Gaussian \eqref{eq:reduced-single}, $\ell_s(\bm\vartheta) = \log \mathcal{N}_p\bigl(\mathbf{y}_s; \bm\mu(\bm\vartheta), \bm\Sigma(\bm\vartheta)\bigr)$, with no compound-symmetric covariance to decompose.
Proposition~\ref{prop:harmonic-mean} reduces to the familiar observation-level conditional predictive ordinate $p(\mathbf{y}_s \mid \mathbf{y}_{-s}) = \bigl(\E_{\bm\vartheta \mid \mathbf{y}}[\,1/p(\mathbf{y}_s \mid \bm\vartheta)\,]\bigr)^{-1}$, and the Taylor scores of Theorem~\ref{prop:laplace-1st} apply verbatim with $\ell_j, \mathbf{s}_j, \mathbf{H}_j$ replaced by their subject-level counterparts $\ell_s$, $\mathbf{s}_s := \nabla \ell_s(\bm\vartheta^\ast)$, and $\mathbf{H}_s := \nabla^2 \ell_s(\bm\vartheta^\ast)$,
\begin{align*}
\text{[first order]} \quad \log \hat{p}^{(1)}(\mathbf{y}_s \mid \mathbf{y}_{-s}) &= \ell_s(\bm\vartheta^\ast) - \tfrac{1}{2}\, \mathbf{s}_s^\top \bm\Omega \mathbf{s}_s \\
\text{[second order]} \quad \log \hat{p}^{(2)}(\mathbf{y}_s \mid \mathbf{y}_{-s}) &= \ell_s(\bm\vartheta^\ast) - \tfrac{1}{2}\, \mathbf{s}_s^\top \mathbf{A}_s^{-1} \mathbf{s}_s + \tfrac{1}{2}\, \log\bigl|\mathbf{I} + \bm\Omega \mathbf{H}_s\bigr|,
\end{align*}
with $\mathbf{A}_s = \bm\Omega^{-1} + \mathbf{H}_s$, the second-order form again requiring $\mathbf{A}_s \succ 0$.
The conditioning of Section~\ref{sec:restrictions} applies unchanged as well, so a single encompassing fit scores every candidate submodel at the subject level.
The same two requirements are met by other designs.
Multigroup data only re-index $\bm\vartheta$, leaving the deletion index and the factorisation over subjects untouched, and jointly modelled exogenous covariates enter the per-unit density as further coordinates (Appendix~\ref{app:covariates}).
In each case it is enough to identify the density of the unit one means to predict; the score and Hessian machinery then applies unchanged.

\section{Simulation Studies}
\label{sec:sim}

We validate the closed-form Taylor elpd scores against brute-force refits to address two questions.
First, does the Taylor score recover the gold-standard predictive density, and where they differ, does the discrepancy originate in the truncation itself or in the posterior approximation on which it is built?
Second, when a submodel is scored by conditioning the encompassing Laplace summary rather than by its own fit, how far may the constraint sit from the summary centre before the conditioned score degrades?

\subsection{Accuracy, Error Decomposition, and Computational Cost}
\label{sec:sim-accuracy}

The study uses $p = 6$ indicators with intercepts $\bm\nu = (0.5, 1.0, 1.5, 0.8, 1.2, 0.6)$, within-level loadings $\bm\lambda^{\mathrm{w}} = (1.0, 0.8, 1.2, 0.9, 1.1, 0.7)$ on a single factor of variance $\psi^{\mathrm{w}} = 1$ with residual variances $\bm\theta^{\mathrm{w}} = (0.5, 0.6, 0.4, 0.7, 0.5, 0.6)$, and a single between-level factor with loadings $\bm\lambda^{\mathrm{b}} = (1.0, 1.1, 0.9, 1.2, 0.8, 1.0)$, variance $\psi^{\mathrm{b}} = 0.5$ and residual variances $\bm\theta^{\mathrm{b}} = (0.20, 0.15, 0.25, 0.10, 0.20, 0.15)$.
Between-level loadings differ from within-level ones, so the generating model is deliberately not invariant across levels.
These values imply standardised within-level loadings of about $0.67$--$0.89$ and indicator intraclass correlations (ICCs) of about $0.23$--$0.38$, a well-identified measurement model with a moderately high ICC.
The error magnitudes reported below should be read against that design, since what a cluster carries about the between-level parameters scales with its ICC.

The theorems of Section~\ref{sec:laplace} are stated for a generic deletion unit, so the marginal and conditional targets differ only in the integrand substituted for $\ell_u$.
We therefore study two designs, one for each target.
Marker-variable scaling is used throughout, and under exact specification, the reported errors purely reflect approximation.
For the marginal target (\loco), we simulate two-level data with distinct within- and between-level one-factor structures, varying the number of clusters $J \in \{30, 100, 200\}$ and cluster sizes $n_j \in \{5, 20\}$, so that $d = 30$. 
Crossing $n_j$ with $J$ aims to isolate the effect of the cluster count from the total sample size.
For the subject-level target (\loso), we draw \emph{unclustered} observations $\mathbf{y}_s \sim \mathcal{N}_p(\bm\nu, \bm\Sigma_{\mathrm{w}})$ from the within-level measurement model and fit the matching single-level confirmatory factor analysis at $N \in \{100, 500, 1000\}$, so that $d = 18$.

Each cell is replicated 100 times.
Within a replicate, a single \texttt{INLAvaan} fit supplies the Taylor scores for every unit. We then execute brute-force \texttt{INLAvaan}-Laplace and \texttt{blavaan}-MCMC refits on $K$ randomly chosen held-out units ($K = 30$ clusters for \loco{}, $K = 100$ subjects for \loso{}, with one MCMC chain of $800$ draws after $400$ burn-in), aligning and pooling the per-unit quantities across replicates. 
The two baselines isolate the two approximations at work.
The short chains are a concession to scale, as the design calls for 48,000 \texttt{blavaan} refits in total.
They still resolve the benchmark finely enough, since each tabulated entry pools $3{,}000$ (\loco) or $10{,}000$ (\loso) of them, and the Gaussian error falls to $0.000$--$0.001$ at the larger designs, which appreciable Monte Carlo noise in the benchmark would have kept above zero.
All \loso{} replicates completed and $96\%$ of \loco{} replicates did, with a further $0.02\%$ of scored units discarded, in both cases predominantly from divergent MCMC refits.
Errors are reported in elpd units.

Figure~\ref{fig-sim-agree} plots the first- and second-order Taylor elpd scores against the MCMC elpd across identical held-out units. 
The second-order estimate tracks the identity line closely across both arms and all sample sizes. As reported in Table~\ref{tbl-sim-decomp}, per-cell discrepancies vanish almost entirely once $J \ge 100$ (\loco) or $N \ge 500$ (\loso). 
The only visible departure occurs at the smallest designs ($J = 30, n_j = 5$ for \loco; $N = 100$ for \loso), and strictly in the first-order series, where the linear term alone fails to capture the curvature of the held-out predictive density. The second-order correction reliably eliminates this bias.

\begin{figure}[htbp]
\centering
\includegraphics[width=\linewidth]{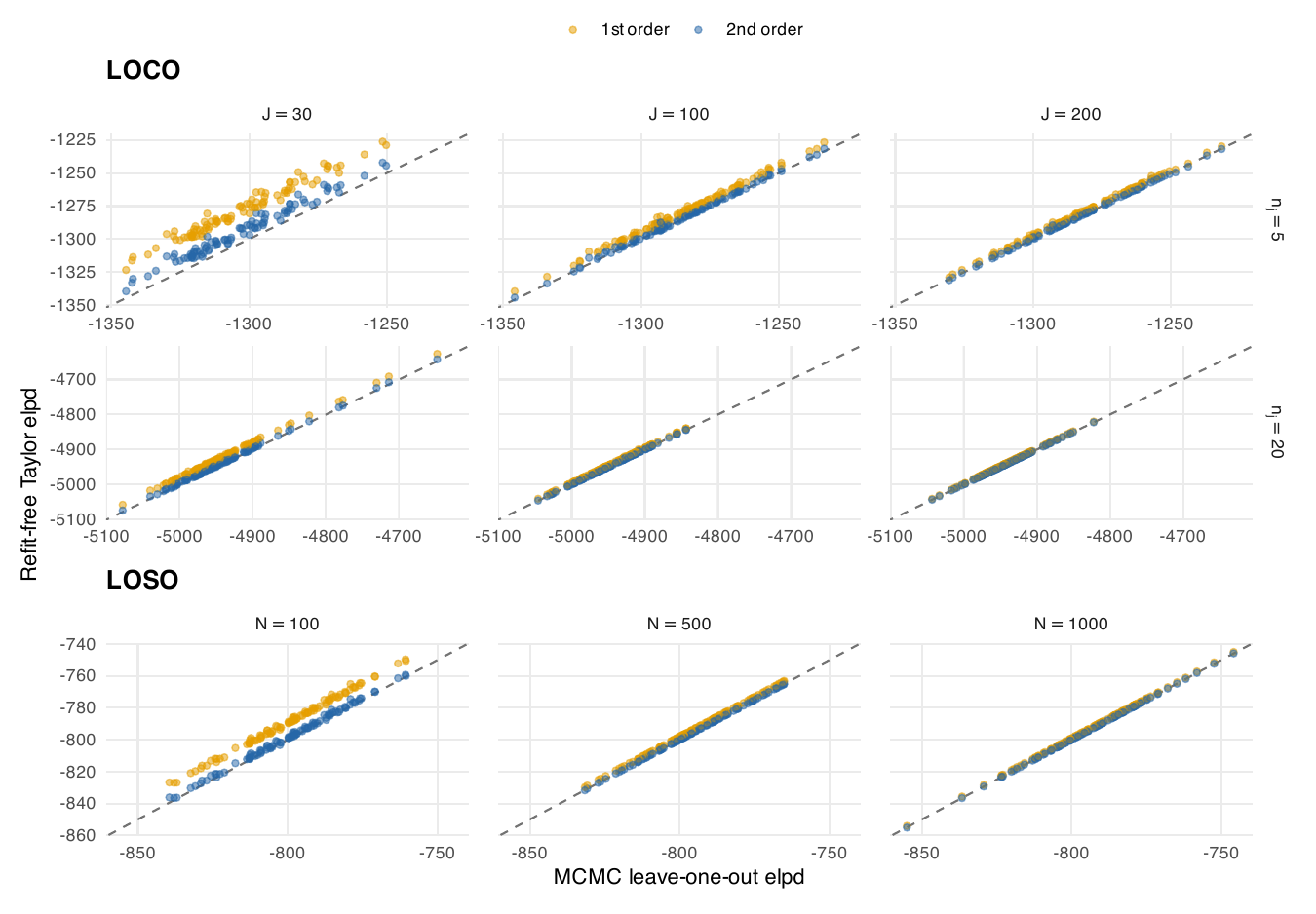}
\caption{First-order (yellow) and second-order (blue) Taylor elpd scores, each from a single fit, against brute-force MCMC-refit elpd over the same held-out units. Each point is one simulated dataset, summing the $K$ refitted units ($K = 30$ clusters, $K = 100$ subjects). The axes are therefore common within each row, since the magnitude of such a sum is set by $n_j$ rather than by $J$. The dashed line indicates exact agreement.}
\label{fig-sim-agree}
\end{figure}

Figure~\ref{fig-sim-agree} also appears to show the first-order bias shrinking as $J$ (\loco{}) or $N$ (\loso{}) grows, but Equation~\eqref{eq:peff-gap} dictates a nonzero limit.
The explanation is that only $K$ units are refitted, and $K$ stays fixed as the design grows.
Each panel therefore shows $K/J$ (or $K/N$) of a total that does not itself shrink, which is why the bias appears to fade.
Figure~\ref{fig-sim-scaling} sums over all units instead, and there the first-order error does not vanish but settles on the limit of \eqref{eq:peff-gap}.
The plotted total error carries the Gaussian error alongside the truncation error, but that error, the prior shrinkage $\tfrac12 \operatorname{tr}(\bm\Omega \mathbf{P})$, and the remainder of \eqref{eq:peff-gap} are all $O(J^{-1})$, leaving $d/2$ as the asymptote: $24.8$ falls to $15.8$ elpd against $d/2 = 15$ for \loco{} and $10.8$ to $9.2$ against $d/2 = 9$ for \loso{}, while the second-order errors appear to diminish.

\begin{figure}[htbp]
\centering
\includegraphics[width=\linewidth]{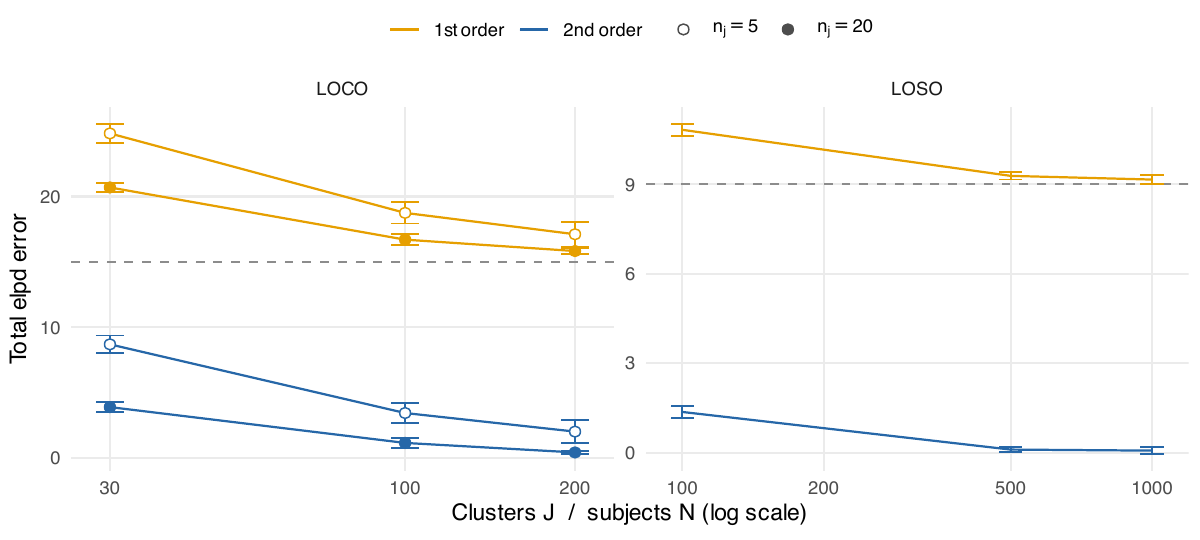}
\caption{Total elpd error over \emph{all} units against the number of clusters $J$ (\loco{}) or subjects $N$ (\loso{}), estimated as $J$ (or $N$) times the mean per-unit error and shown with $\pm 2$ sampling standard errors. The dashed line is the asymptotic $d/2$ limit of \eqref{eq:peff-gap}, which the summed gap approaches to $O(J^{-1})$.}
\label{fig-sim-scaling}
\end{figure}

To isolate the source of any residual discrepancy, we introduce a third quantity: a brute-force \texttt{INLAvaan} refit on the reduced data $\mathbf{y}_{-u}$, with $u=j$ for clusters and $u=s$ for subjects. 
This allows us to partition the total error into two distinct components: 
\begin{equation*}
\underbrace{\widehat{\mathrm{elpd}}_u - \mathrm{elpd}^{\mathrm{MCMC}}_u}_{\text{total error}}
\;=\;
\underbrace{\widehat{\mathrm{elpd}}_u - \mathrm{elpd}^{\mathrm{INLA}}_u}_{\text{truncation error}}
\;+\;
\underbrace{\mathrm{elpd}^{\mathrm{INLA}}_u - \mathrm{elpd}^{\mathrm{MCMC}}_u}_{\text{Gaussian error}}.
\end{equation*}
The truncation error compares the Taylor score against an actual refit under the identical Laplace posterior, isolating the behaviour of the expansion itself. The Gaussian error compares that \texttt{INLAvaan} refit against MCMC, isolating the limits of the Gaussian-posterior approximation, which, by Bernstein--von Mises arguments, should vanish as information accrues. 

Table~\ref{tbl-sim-decomp} details this partition, and two conclusions follow.
First, at second order the Taylor truncation contributes negligibly to the overall discrepancy. 
Its second-order error is small, largest at $0.061$ elpd ($J = 30, n_j = 5$), falling to $0.002$ or below for \loco~at $J \ge 200$ and to nearly zero for \loso~at $N \ge 500$, and remains well below the first-order error (which reaches $0.60$) throughout, confirming that the curvature correction sharpens the estimate without over-correcting. 
Because this component compares the Taylor score against a true refit, the residual error against MCMC is not attributable to our expansion machinery. 
Summed over units, the first-order truncation error runs from $18.0$ down to $15.6$ elpd across the six \loco{} cells and from $9.9$ to $9.2$ across the three \loso{} cells, against second-order sums that stay below $1.9$ throughout.
Measured directly using $\tfrac12 \sum_j \operatorname{tr}(-\bm\Omega \mathbf{H}_j)$, $\tfrac12 p_{\mathrm{eff}}$ is $8.6$--$9.0$ against $d/2 = 9$ for \loso{} and $13.4$--$14.9$ against $d/2 = 15$ for \loco{}, the shortfall in each case being the prior shrinkage $\tfrac12 \operatorname{tr}(\bm\Omega \mathbf{P})$ of \eqref{eq:peff-gap}. 
The gap itself behaves as \eqref{eq:peff-gap} dictates: it exceeds $\tfrac12 p_{\mathrm{eff}}$ in every cell by the $O(J^{-1})$ remainder, and this remainder falls from $2.7$ elpd at $J = 30$ to $0.5$ at $J = 200$ for \loco{}, and from $0.9$ to $0.1$ for \loso{}.

Second, the inherited Gaussian posterior approximation error vanishes predictably.
For \loco{} the decay is governed by the number of clusters $J$ rather than the total sample size: the error runs $0.09$--$0.23$ elpd per cluster at $J = 30$ and falls to $0.025$ or below once $J \ge 100$.
This aligns with the large-sample behaviour expected of \loco{}, where each cluster contributes a single between-level observation.
No fit-level quantity we examined predicted the error better than the design itself.
The magnitudes are specific to this design, but the pattern is not: the Gaussian error is largest where $J$ is smallest, which is also where a deleted refit is cheapest to run as a check.

\begin{table}[tbp]

\caption{\label{tbl-sim-decomp}Per-unit mean error decomposition by design cell (elpd units; 100 replicates per cell). \emph{1st \& 2nd order} is the Taylor score against an \texttt{INLAvaan} refit; \emph{Gaussian} is that refit against a full MCMC refit; \emph{Total} is their sum. \emph{LOO time} is the wall-clock cost of the Taylor pass over all units. Monte Carlo standard errors do not exceed $0.002$ in the first- and second-order columns and $0.012$ in the Gaussian and Total columns.}

\centering{

\fontsize{12.0pt}{14.0pt}\selectfont
\begin{tabular*}{\linewidth}{@{\extracolsep{\fill}}rrrrrrr}
\toprule
 &  & \multicolumn{4}{c}{{Mean per-unit error (elpd)}} &  \\ 
\cmidrule(lr){3-6}
$N/J$ & $n_j$ & 1st order & 2nd order & Gaussian & Total & LOO time (s) \\ 
\midrule\addlinespace[2.5pt]
\multicolumn{7}{c}{Leave-one-subject-out (LOSO)} \\[2.5pt] 
\midrule\addlinespace[2.5pt]
100 & -- & 0.099 & 0.005 & 0.009 & 0.014 & 0.140 \\ 
500 & -- & 0.019 & 0.000 & 0.000 & 0.000 & 0.185 \\ 
1000 & -- & 0.009 & 0.000 & 0.000 & 0.000 & 0.249 \\ 
\midrule\addlinespace[2.5pt]
\multicolumn{7}{c}{Leave-one-cluster-out (LOCO)} \\[2.5pt] 
\midrule\addlinespace[2.5pt]
30 & 5 & 0.598 & 0.061 & 0.229 & 0.290 & 0.325 \\ 
30 & 20 & 0.596 & 0.037 & 0.093 & 0.130 & 0.336 \\ 
100 & 5 & 0.162 & 0.009 & 0.025 & 0.034 & 0.474 \\ 
100 & 20 & 0.160 & 0.004 & 0.007 & 0.012 & 0.495 \\ 
200 & 5 & 0.078 & 0.002 & 0.008 & 0.010 & 0.743 \\ 
200 & 20 & 0.078 & 0.001 & 0.001 & 0.002 & 0.659 \\ 
\bottomrule
\end{tabular*}

}

\end{table}

The closed-form procedure is orders of magnitude faster than brute-force cross-validation.
At the largest design point ($J = 200$, $n_j = 20$), one \texttt{INLAvaan} fit plus the closed-form pass over all $200$ clusters takes $2.1$ seconds, against $2.0$ seconds per deleted-cluster \texttt{INLAvaan} refit ($406$ seconds in total, $193\times$) and $24.1$ seconds per \texttt{blavaan}-MCMC refit ($4{,}814$ seconds, $2{,}286\times$).
That ratio widens with $J$, and the reduction in cost entails no loss of accuracy, since the second-order Taylor score agrees with the MCMC benchmark to within $0.002$ elpd per cluster at this design point.

\subsection{Accuracy of Conditioned Submodel Scores}
\label{sec:sim-conditioning}

The compatible-prior identity of Section~\ref{sec:restrictions} is exact, so scoring a submodel by conditioning carries only one further approximation.
The constraint is applied to the Gaussian Laplace summary rather than to the posterior itself, and the quality of that local step must depend on how far the constraint sits from the centre.
To measure that dependence directly, consider that for a restriction $\mathbf{C}_m\bm\vartheta = 0$ of rank $r$, the natural measure of that distance is the Mahalanobis distance from the centre to the constraint surface,
\begin{equation}
D_m^2 \;:=\; (\mathbf{C}_m\bm\vartheta^\ast)^\top \bigl(\mathbf{C}_m \bm\Omega \mathbf{C}_m^\top\bigr)^{-1} (\mathbf{C}_m\bm\vartheta^\ast),
\label{eq:cond-distance}
\end{equation}
read from the Laplace summary $\mathcal{N}(\bm\vartheta^\ast, \bm\Omega)$ that \eqref{eq:restr-mean}--\eqref{eq:restr-cov} condition.
Two readings make $D_m$ interpretable.
It is the size of the mean shift that conditioning induces, since \eqref{eq:restr-mean} gives $(\bm\vartheta^\ast_m - \bm\vartheta^\ast)^\top \bm\Omega^{-1} (\bm\vartheta^\ast_m - \bm\vartheta^\ast) = D_m^2$.
Under the encompassing Gaussian $D_m^2$ follows a $\chi^2_r$ law, so $\Pr(\chi^2_r > D_m^2)$ is the posterior mass lying further from the centre than the constraint itself.
For a single dropped path, $D_m$ reduces to $|\vartheta^\ast_a| / \sqrt{\Omega_{aa}}$, the number of posterior standard deviations between the fitted coefficient and the zero it is constrained to.

For this study, two-level data are generated with $p = 12$ indicators loading on four three-indicator factors at each level: three latent predictors $x_1, x_2, x_3$ and an outcome $y$.
Standardised within-level loadings cycle through $(0.80, 0.70, 0.90)$ and between-level loadings through $(0.75, 0.85, 0.70)$, indicator intercepts through $(0.5, 1.0, 1.5)$, and every indicator has an ICC of $0.20$.
At the between level the outcome is regressed on all three predictors, $y_{\mathrm{b}} = \beta_1 x_{1b} + \beta_2 x_{2b} + \beta_3 x_{3b} + \zeta_{\mathrm{b}}$, with $(\beta_1, \beta_2, \beta_3) = (0.45, 0.12, 0)$, i.e., one strong path, one weak, and one exactly null.
The three latent predictors are mutually correlated at $\rho$, and we take $\rho \in \{0.35, 0.80\}$ to vary how strongly their coefficient estimates are entangled.
The candidate set is all $2^3 = 8$ subsets of these $\beta$-paths, so restrictions of every rank $r \in \{1, 2, 3\}$ occur.
We use $J \in \{30, 100\}$ clusters of size $n_j = 20$, with $150$ replicates per cell.

For each dataset the encompassing model is fitted once and every candidate is scored twice, first by conditioning the encompassing Laplace summary on $\mathbf{C}_m\bm\vartheta = 0$, and then by a separate \texttt{INLAvaan} fit of the candidate itself.
The two routes share the posterior machinery, so their difference isolates the conditioning step alone.
Not every replicate yielded a complete set of candidate scores, and the retained counts are given in Table~\ref{tbl-sim-cond} and the attrition described alongside it.
Figure~\ref{fig-sim-cond} plots the conditioned minus separately fitted \loco{} elpd against $D_m$, over the $3{,}661$ restricted candidate scorings.
Sampling variability spreads $D_m$ from zero to $6.7$ posterior standard deviations.
The pattern is the same in all four design cells.
Near the centre the error is small and symmetric about zero, and beyond two standard deviations it turns systematically negative and grows.

\begin{figure}[htbp]
\centering
\includegraphics[width=0.95\linewidth]{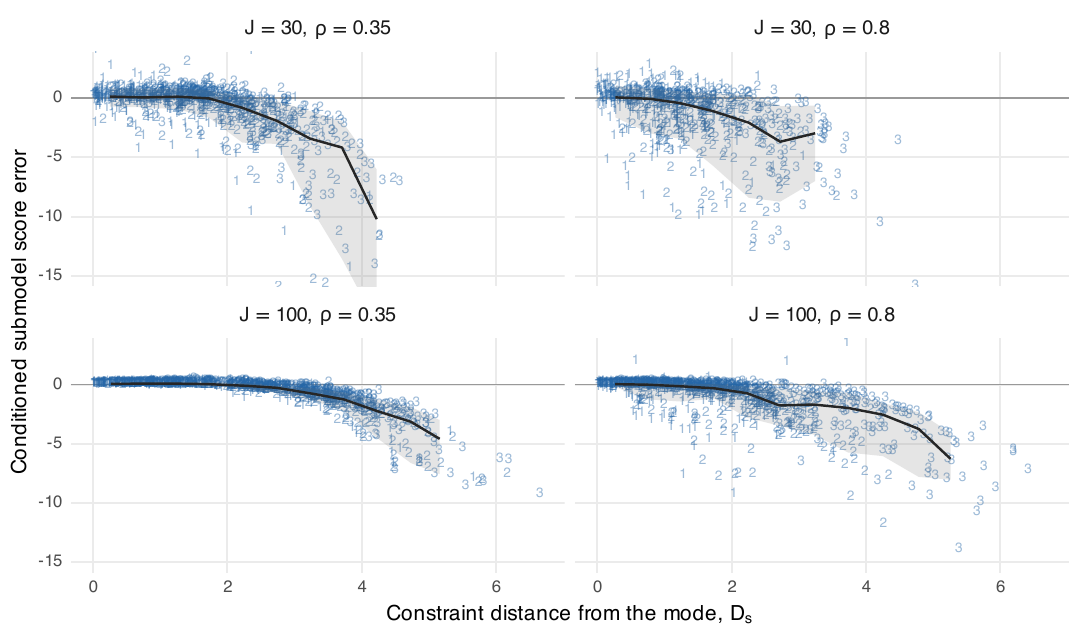}
\caption{Error of the conditioned submodel score (conditioned minus separately fitted elpd) against the constraint's Mahalanobis distance $D_m$ from the encompassing summary centre, by design cell. Each point is plotted as the number of paths that candidate restricts. The solid line is the median error in half-standard-deviation bins and the band spans the 10th--90th percentile.}
\label{fig-sim-cond}
\end{figure}

Three features stand out.
First, the one-sidedness sharpens with distance.
Pooling the four design cells, the error is negative in $38\%$ of scorings within one standard deviation of the centre, in $59\%$ between one and two, $90\%$ between two and three, $99\%$ between three and four, and in all $225$ scorings beyond four.
The median error over the same bands runs $+0.04$, $-0.07$, $-0.66$, $-1.26$ and $-3.58$ elpd.
A submodel conditioned far from the centre is therefore \emph{under-scored} rather than over-scored, which rules out the spurious promotion of a heavily restricted candidate, though not its false exclusion.
Second, the error is governed by $D_m$ and not by the rank of the restriction.
The median absolute error in the $D_m \le 1$ band is $0.13$, $0.15$ and $0.33$ elpd for $r = 1, 2, 3$, and in the $D_m > 4$ band $4.73$, $3.35$ and $3.57$, so the distance matters far more than the number of paths being constrained.
Third, $D_m$ is read from the encompassing Laplace summary alone, so the candidates at risk are identifiable before any scoring is done.

$D_m$ is measured in units of $\bm\Omega$, so it places the candidates of a single encompassing fit on a common scale.
That scale does not carry across fits, however.
Correlation among the latent predictors inflates the posterior standard deviations of the individual coefficients, so at $\rho = 0.80$ the constraints sit \emph{closer} to the centre---a median $D_m$ of $1.55$ against $2.37$ at $\rho = 0.35$ when $J = 100$---and yet the error is larger in every distance band, reaching a median absolute $3.05$ elpd against $1.64$.
Conditioning displaces the summary along the flat ridge that collinearity creates, and a quadratic (Laplace) approximation is least accurate along exactly such directions.

These errors matter only insofar as they change what a comparison concludes, so the relevant question is whether the conditioned scores lead to the same decision as refitting.
At $J = 100$ the conditioned scores reproduce the ranking that separate fits produce almost exactly (Table~\ref{tbl-sim-cond}), with a Kendall correlation of $0.95$ at $\rho = 0.35$ and $0.85$ at $\rho = 0.80$, and they select the same leading candidate in $95\%$ and $87\%$ of replicates.
Where the leaders differ, the disagreement is between near-ties. 
The \emph{predictive regret} from selecting the highest-scoring candidate, the elpd forfeited relative to the best under separate fits, is $0.004$ and $0.041$.
A $\Delta\mathrm{elpd}$ between two candidates carries a paired standard error, computed from the per-cluster differences in log-CPO contributions in the usual normal approximation \citep{vehtari2017practical}.
The conditioning error amounts to $0.27$ and $0.52$ of that standard error at $J = 100$, so it is well inside the noise the comparison already carries.
At $J = 30$ every one of these degrades, and the conditioning error becomes comparable to the paired standard error itself.
The real application of Section~\ref{sec:midus} shows the same behaviour, its near-centre candidates being the least under-scored of the 16.

\begin{table}[tbp]
\caption{\label{tbl-sim-cond}Concordance between conditioned submodel scores and independent \texttt{INLAvaan} fits.
Designs are ordered from hardest to easiest (Monte Carlo SEs in parentheses). \emph{Reps} is the number of replicates, of $150$, with a complete set of candidate scores; the attrition is described in the text. \emph{Whole ranking} is Kendall's $\tau$ between candidate orderings; \emph{best model} is the proportion of top-rank matches. \emph{Elpd lost} is the predictive regret from conditioned selection; \emph{in SE units} expresses it relative to the paired standard error, so a value below $1$ lies within estimation noise.}
\centering
\begin{tabular*}{\linewidth}{@{\extracolsep{\fill}}rrrcccc}
\toprule
 & & & \multicolumn{2}{c}{Agrees with refitting} & \multicolumn{2}{c}{Size of the disagreement} \\
\cmidrule(lr){4-5}\cmidrule(lr){6-7}
$J$ & $\rho$ & reps & whole ranking & best model & elpd lost & in SE units \\
\midrule
$30$ & $0.80$ & $98$ & $0.57\ (0.03)$ & $0.53\ (0.05)$ & $0.379\ (0.069)$ & $1.17\ (0.10)$ \\
$30$ & $0.35$ & $125$ & $0.72\ (0.02)$ & $0.71\ (0.04)$ & $0.139\ (0.029)$ & $0.89\ (0.14)$ \\
$100$ & $0.80$ & $150$ & $0.85\ (0.01)$ & $0.87\ (0.03)$ & $0.041\ (0.019)$ & $0.52\ (0.03)$ \\
$100$ & $0.35$ & $150$ & $0.95\ (0.01)$ & $0.95\ (0.02)$ & $0.004\ (0.002)$ & $0.27\ (0.02)$ \\
\bottomrule
\end{tabular*}
\end{table}

Two limitations should be read alongside these figures.
At $J = 30$ a complete set of candidate scores was obtained in $125$ of $150$ replicates at $\rho = 0.35$ and $98$ of $150$ at $\rho = 0.80$.
The remainder lost at least one cluster to a non-positive-definite $\mathbf{A}_j$, and did so in the separately fitted references about as often as in the conditioned scores.
The retained replicates are thus the better-behaved ones, and the $J = 30$ summaries should be read as mildly optimistic.
The reference itself is also least reliable there: in five of the $J = 30$ candidate cells a separate fit failed to converge, affecting the conditioned and separately fitted routes identically.

\section{Cluster-Level Model Choice in PISA School-Safety Data}\label{sec:pisa}

Consider the PISA 2022 student questionnaire \autocite{oecd2024pisa}, restricted to Saudi Arabia, with $N = 6{,}125$ students in $J = 186$ schools (sizes $10$--$42$, median $36$).
Two latent constructs describe the school-safety climate, each measured on \emph{both} levels: (a) \textbf{perceived safety} (four items---feeling safe to and from school, in classrooms, and elsewhere) and (b) \textbf{bullying victimisation} (the six-item PISA bullying scale).
An exploratory fit with the level-2 loadings freely estimated terminated on the boundary of the parameter space, with a near-singular between-school covariance matrix.
The ten indicators carry ICCs of only $0.005$--$0.038$, so there is enough between-school information to support common loadings but not a freely loaded level-2 measurement model.
We therefore constrain the loadings of each construct equal across the two levels.
Two student-level covariates, a student's sense of belonging and socio-economic status (SES), are modelled jointly with the responses, so each splits into a within- and a between-school part and freely covaries with the factors.
What we set out to learn is whether their \emph{between-school} part predicts that school's safety climate.

We answer this by comparing three nested models that share the same variables and within-school structure and differ \emph{only} in which between-school paths onto perceived safety are freed.
The baseline $\mathcal{M}_1$ is a pure random intercept, leaving school safety orthogonal to the predictors.
$\mathcal{M}_2$ frees the school-level belonging-climate path, regressing perceived safety on school-mean belonging, and $\mathcal{M}_3$ adds a contextual SES path on top.
Each consecutive \loco~difference then measures how much a school's belonging climate, or its socio-economic composition, sharpens prediction of safety in a school never seen during fitting.
Models are compared by the difference of their summed \loco~scores, $\Delta\mathrm{elpd}$, together with the paired standard error of Section~\ref{sec:sim-conditioning}, here computed from the per-school differences in log-CPO contributions.

The Taylor elpd scores in Table~\ref{tbl-pisa} come from three single fits, each taking $43$ seconds inclusive of the \loco~pass over all $186$ schools.
The criterion keeps the belonging path, as freeing the latent regression of belonging to perceived safety ($\mathcal{M}_2$) gives the highest leave-one-school-out elpd, $17.4$ above the random-intercept baseline.
A school where students feel they belong is, predictably, a school where students \emph{feel safe}, and that link improves out-of-cluster prediction.
Adding a contextual SES path on top ($\mathcal{M}_3$) does not help, and the elpd \emph{falls} by $0.7$ well inside its paired standard error, so the two models are not separated on this criterion.
School-mean SES carries little school-level signal here (its between-school correlation with safety is only $\approx 0.14$), so there is little for the marginal criterion to reward.

\begin{table}[tbp]

\caption{\label{tbl-pisa}\loco~model comparison on the PISA 2022 school-safety data, Saudi Arabia ($J = 186$ schools, $N = 6{,}125$ students), by second-order Taylor elpd scores and by MCMC$+$PSIS. Higher elpd means better prediction of a held-out school. $\Delta\elpd$ is the difference from the random-intercept baseline $\mathcal{M}_1$, with paired cluster-level standard error from the per-school differences. The best value in each column is set in bold.}

\centering{

\setlength{\tabcolsep}{5pt}
\begin{tabular*}{\linewidth}{@{\extracolsep{\fill}}llrrrrr}
\toprule
& School-level & & \multicolumn{2}{c}{$\elpd_{\loco}$} & \multicolumn{2}{c}{$\Delta\elpd$ (SE)} \\
\cmidrule(lr){4-5}\cmidrule(lr){6-7}
& predictors of safety & $d$ & Taylor & MCMC$+$PSIS & Taylor & MCMC$+$PSIS \\
\midrule\addlinespace[2.5pt]
$\mathcal{M}_1$ & none & 58 & $-61214.6$ & $-61217.8$ & $-$ & $-$ \\
$\mathcal{M}_2$ & $+$ belonging & 59 & $\bm{-61197.2}$ & $-61201.0$ & $\bm{17.4}$ (6.2) & $16.8$ (6.8) \\
$\mathcal{M}_3$ & $+$ belonging, SES & 60 & $-61197.9$ & $\bm{-61198.0}$ & $16.7$ (6.2) & $\bm{19.7}$ (6.6) \\
\bottomrule
\end{tabular*}

}

\end{table}%

These \loco~scores can be checked against the route an applied user would otherwise take.
For a two-level model \texttt{blavaan}'s Stan target returns a log-likelihood with one entry per cluster, the level-2 effects already integrated out \autocite{merkle2021efficient}, so Pareto-smoothed importance sampling applied to that matrix estimates the same $\elpd_{\loco}$ from a single posterior sample.
The cross-level equality constraint is not expressible through \texttt{blavaan}'s modelling interface, so we impose it in the Stan program it generates, aliasing the free level-2 loading vector to its level-1 counterpart. 
Every other prior is left at defaults identical to \texttt{INLAvaan}'s, so that what is compared is two approximations to one estimand rather than two models.
Fitting all three models with four chains of 2,000 warm-up and 2,000 sampling iterations, each completing roughly under ten minutes, gives the \emph{MCMC+PSIS} columns of Table~\ref{tbl-pisa}.
Every sampler converged cleanly, no divergent transitions were reported, and all $\widehat{R} \le 1.003$.
The two routes agree to within $3.8$ elpd throughout, and the paired standard errors are comparable, so the Taylor scores are validated against the established route.

The two routes do part company on the SES path, where the Taylor score puts $\mathcal{M}_2$ ahead of $\mathcal{M}_3$ by $0.7$ elpd and the MCMC column leans the other way by $3.0$, and neither margin clears twice its paired SE.
In this comparison, however, MCMC$+$PSIS proves inadequate, as deleting an entire school strains the importance weights regardless of how well the sampler has converged.
Figure~\ref{fig-pisa-mcmc} illustrates this instability.
The two methods agree wherever the Pareto index $\hat k_j$---the tail diagnostic attached to each smoothed weight \autocite{vehtari2024pareto}---is small. 
Conversely, the largest discrepancies occur among the 7 out of 558 school--model pairs that PSIS flags as unreliable, with the two most extreme cases appearing in the $\mathcal{M}_2$ score where $\hat k_j>1$.
The \texttt{loo\_compare()} function \autocite{sivula2025uncertainty} from the R \texttt{loo} package warns as much, and flags the difference as too small for its standard error to be trusted.
By contrast, the validity check \eqref{eq:existence-2nd} for the Taylor scores passed (the largest eigenvalue reaches 0.13, against a validity boundary of 1), handing us a clear basis for decision-making.

\begin{figure}[htbp]

\centering{

\includegraphics[width=0.95\linewidth,height=\textheight,keepaspectratio]{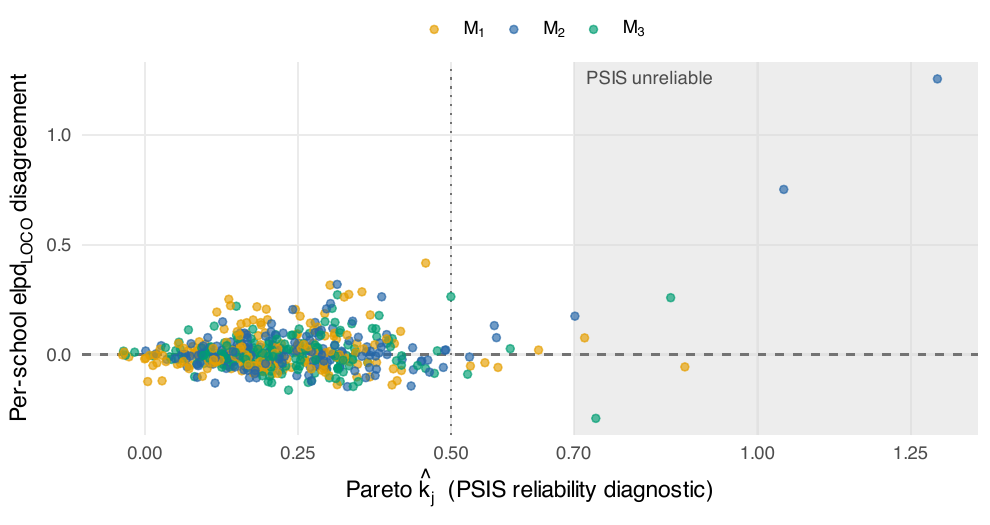}

}

\caption{\label{fig-pisa-mcmc}Per-school disagreement between the Taylor $\elpd_\loco$ and its MCMC$+$PSIS counterpart, against the PSIS reliability diagnostic $\hat k_j$, for all three models. Below $\hat k_j = 0.5$ the discrepancy is centred on zero. The seven pairs in the shaded region beyond $\hat k_j = 0.7$, where the importance-sampling estimate is not to be trusted, include the two largest discrepancies, both in the $\mathcal{M}_2$ score at $\hat k_j > 1$.}

\end{figure}%

\section{Predictive Structure Selection in the MIDUS Sibling Sample}
\label{sec:midus}

We use the sibling sample of the MIDUS study \autocite{radler2014midlife} to ask whether the long-run prospective associations of Extraversion (E) and Openness (O) with psychological well-being (PWB) are best represented as personality preceding PWB, PWB preceding personality, reciprocal prediction, or temporal stability alone.
The restriction to E and O is substantive rather than data-driven, as the two traits form the Plasticity metatrait in higher-order models of personality \autocite{digman1997higherorder,deyoung2015cybernetic}. 
Furthermore, an earlier analysis of the same three MIDUS waves found that E and O, but not the other Big Five traits, predicted subsequent PWB \autocite{joshanloo2023reciprocal}.
The waves were collected in 1995--1996, 2004--2006, and 2013--2014, so the two intervals are of approximately equal length.
For respondent $s$ in sibling-design family $j$ and wave $t$, let $E_{sjt}$ and $O_{sjt}$ denote the released Midlife Development Inventory \autocite[MIDI;][]{lachman1997midlife} scale summaries and let $\mathbf{r}_{sjt}=(r_{sj1t},\ldots,r_{sj6t})^\top$ contain the six Ryff domain scores: autonomy, environmental mastery, personal growth, positive relations, purpose in life, and self-acceptance \autocite{ryff1995structure}.
Stacked over the three waves, the jointly modelled response vector is
\begin{equation*}
\mathbf{y}_{sj}
=
\bigl(E_{sj1},O_{sj1},\mathbf{r}_{sj1}^{\top},
      E_{sj2},O_{sj2},\mathbf{r}_{sj2}^{\top},
      E_{sj3},O_{sj3},\mathbf{r}_{sj3}^{\top}\bigr)^{\top}
\in\mathbb R^{24},
\end{equation*}
and each score is placed on a common Wave-1 reference scale, using the same centring and scaling constants at all three waves.
Complete observations give $N=761$ respondents in $J=429$ families of size one to seven, of which 204 are singletons after longitudinal filtering.
In the notation of Section~\ref{sec:setting}, then, $p=24$ and $n_j\in\{1,\ldots,7\}$.

Every coordinate is decomposed into a family and a respondent component in the manner of Section~\ref{sec:setting}, with $d=1,\ldots,6$ indexing the Ryff domains:
\begin{equation*}
E_{sjt}=\nu_{Et}+E_{jt}^{\mathrm{b}}+E_{sjt}^{\mathrm{w}},
\qquad
O_{sjt}=\nu_{Ot}+O_{jt}^{\mathrm{b}}+O_{sjt}^{\mathrm{w}},
\qquad
r_{sjdt}=\nu_{dt}+r_{jdt}^{\mathrm{b}}+r_{sjdt}^{\mathrm{w}},
\end{equation*}
the intercepts free at every item and wave.
Within families, E and O enter as observed summaries with no measurement model of their own, while the six domain components share a single PWB factor $\eta_{sjt}^{\mathrm{w}}$ at each wave, $r_{sjdt}^{\mathrm{w}}=\lambda_d^{\mathrm{w}}\eta_{sjt}^{\mathrm{w}}+\epsilon_{sjdt}^{\mathrm{w}}$.
Self-acceptance is the marker variable, and the remaining loadings are equal over waves \autocite{joshanloo2023reciprocal}, which places $\eta^{\mathrm{w}}_{sjt}$ on a common scale and makes its autoregression interpretable.
Residual variances are free and their covariances are zero, within and across waves.

Collecting the within-family state as $\mathbf{h}_{sjt}^{\mathrm{w}}=(E_{sjt}^{\mathrm{w}},O_{sjt}^{\mathrm{w}},\eta_{sjt}^{\mathrm{w}})^\top$, and writing subscripts outcome-first so that, e.g., $\beta_{\eta E}$ denotes $E_{sjt}^{\mathrm{w}}\to\eta_{sj,t+1}^{\mathrm{w}}$, the transition model is
\begin{equation*}
\mathbf{h}_{sj,t+1}^{\mathrm{w}}
=\bm\Gamma^{\mathrm{w}}(\bm\delta)\mathbf{h}_{sjt}^{\mathrm{w}}+\bm\zeta_{sj,t+1}^{\mathrm{w}},
\qquad t=1,2,
\qquad
\bm\Gamma^{\mathrm{w}}(\bm\delta)=
\begin{bmatrix}
a_E & 0 & \delta_{E\eta}\beta_{E\eta}\\
0 & a_O & \delta_{O\eta}\beta_{O\eta}\\
\delta_{\eta E}\beta_{\eta E}
  & \delta_{\eta O}\beta_{\eta O} & a_\eta
\end{bmatrix},
\end{equation*}
with $\mathbf{h}_{sj1}^{\mathrm{w}}\sim\mathcal N_3(\mathbf{0},\bm\Psi_1^{\mathrm{w}})$ and $\bm\zeta_{sjt}^{\mathrm{w}}\sim\mathcal N_3(\mathbf{0},\bm\Psi_t^{\mathrm{w}})$ for $t=2,3$, the three unrestricted diagonal blocks of the within-level disturbance covariance $\bm\Psi^{\mathrm{w}}$ of Section~\ref{sec:setting}.
The diagonal stability paths are always included, while the four personality--PWB paths are selected, switched on and off by the indicators $\bm\delta$.
Each coefficient is shared across the two intervals, which the near-equal wave spacing makes reasonable, so one $\delta$ controls two displayed arrows.
E--O cross-lags are fixed to zero, their contemporaneous association remaining free at every wave through $\bm\Psi_t^{\mathrm{w}}$, so that the candidates differ along one substantive axis only.

Between families, each of the eight constructs $v\in\{E,O,r_1,\ldots,r_6\}$ has one stable component carrying unit loadings on its three wave scores, $v_{jt}^{\mathrm{b}}=f_{jv}^{\mathrm{b}}+e_{jvt}^{\mathrm{b}}$, with $\mathbf{f}_j^{\mathrm{b}}=(f_{jE}^{\mathrm{b}},\ldots,f_{jr_6}^{\mathrm{b}})^\top\sim\mathcal N_8(\mathbf{0},\bm\Psi^{\mathrm{b}})$ unrestricted and free wave-specific residual variances $\bm\Theta^{\mathrm{b}}$.
This block captures stable family covariance without assigning temporal direction.
It carries no between-family PWB factor and no level-2 paths, which keeps the entire structural question at the respondent-within-family level.

Order the four path indicators as $\bm\delta=(\delta_{\eta E},\delta_{E\eta},\delta_{\eta O},\delta_{O\eta})^\top \in\{0,1\}^4$.
For either personality trait $T\in\{E,O\}$, the pair $(\delta_{\eta T},\delta_{T\eta})$ distinguishes stability alone $(0,0)$, trait-leading $(1,0)$, PWB-leading $(0,1)$, and reciprocal prediction $(1,1)$.
The Cartesian product for E and O gives $2^4=16$ candidates, so the eight optional arrows in Figure~\ref{fig:midus-overview} represent just four decisions, not eight.
The encompassing model has $\bm\delta=(1,1,1,1)^\top$, and every other structure is the linear restriction $\mathbf{C}_{\bm\delta}\bm\vartheta=0$ of Section~\ref{sec:restrictions} that sets the omitted $\beta$ coordinates to zero.
Everything else---autoregressions, covariances, measurement parameters, intercepts, and the family block---remains free and updates when the restriction is imposed.
Since \texttt{INLAvaan} places independent priors on the parameters by default, the compatible-prior convention \eqref{eq:compat} holds for every $\mathbf{C}_{\bm\delta}$.

\begin{figure}[t]
\centering
\begin{tikzpicture}[
  scale=1.1,
  >=Stealth,
  trait/.style={draw,rectangle,minimum size=0.80cm,inner sep=1pt,
                fill=white,font=\small},
  pwb/.style={draw,circle,minimum size=0.78cm,
              fill=white,align=center,font=\small},
  fixed/.style={->,very thick,black},
  forward/.style={->,very thick,dashed,blue!70!black},
  reverse/.style={->,very thick,dash dot,orange!85!black,
                  preaction={draw=white,line width=3.4pt}},
  wave/.style={font=\small\bfseries},
  edgelab/.style={font=\scriptsize,fill=white,inner sep=0.7pt}
]

\foreach \w/\x in {1/0,2/5.4,3/10.8}{
  \node[wave] at (\x,4.55) {Wave \w};
  \node[trait] (E\w) at (\x,3.75) {$E_{\w}^{\mathrm{w}}$};
  \node[pwb]   (P\w) at (\x,2.15) {$\eta_{\w}^{\mathrm{w}}$};
  \node[trait] (O\w) at (\x,0.55) {$O_{\w}^{\mathrm{w}}$};
}

\draw[fixed] (E1) -- node[edgelab,above] {$a_E$} (E2);
\draw[fixed] (E2) -- node[edgelab,above] {$a_E$} (E3);
\draw[fixed] (P1) -- node[edgelab,above] {$a_\eta$} (P2);
\draw[fixed] (P2) -- node[edgelab,above] {$a_\eta$} (P3);
\draw[fixed] (O1) -- node[edgelab,below] {$a_O$} (O2);
\draw[fixed] (O2) -- node[edgelab,below] {$a_O$} (O3);

\draw[forward] (E1.south east) to[out=-8,in=168]
  node[pos=0.28,edgelab,above,sloped,text=blue!70!black]
  {$\beta_{\eta E}$} (P2.north west);
\draw[forward] (E2.south east) to[out=-8,in=168]
  node[pos=0.28,edgelab,above,sloped,text=blue!70!black]
  {$\beta_{\eta E}$} (P3.north west);
\draw[reverse] (P1.north east) to[out=12,in=192]
  node[pos=0.28,edgelab,below,sloped,text=orange!85!black]
  {$\beta_{E\eta}$} (E2.south west);
\draw[reverse] (P2.north east) to[out=12,in=192]
  node[pos=0.28,edgelab,below,sloped,text=orange!85!black]
  {$\beta_{E\eta}$} (E3.south west);

\draw[forward] (O1.north east) to[out=8,in=192]
  node[pos=0.28,edgelab,below,sloped,text=blue!70!black]
  {$\beta_{\eta O}$} (P2.south west);
\draw[forward] (O2.north east) to[out=8,in=192]
  node[pos=0.28,edgelab,below,sloped,text=blue!70!black]
  {$\beta_{\eta O}$} (P3.south west);
\draw[reverse] (P1.south east) to[out=-12,in=168]
  node[pos=0.28,edgelab,above,sloped,text=orange!85!black]
  {$\beta_{O\eta}$} (O2.north west);
\draw[reverse] (P2.south east) to[out=-12,in=168]
  node[pos=0.28,edgelab,above,sloped,text=orange!85!black]
  {$\beta_{O\eta}$} (O3.north west);

\path ([yshift=-4mm]current bounding box.south);

\end{tikzpicture}%
\caption{Compact overview of the encompassing respondent-within-family model. Black stability paths are retained throughout; the coloured $\beta$ paths are selectively set to zero in the submodels. Each coefficient is shared across the two intervals; measurement indicators for $\eta$ are omitted, and the within-wave disturbance covariances among $E$, $\eta$, and $O$ are freely estimated in every candidate but not drawn.}
\label{fig:midus-overview}
\end{figure}

Each family is one cluster, so the cluster reduced form \eqref{eq:cluster-reduced} applies with $p=24$, and we rank the 16 restrictions by their Taylor $\elpd_{\loco}$ scores.
Whole-family deletion asks which structure best predicts a family unseen during fitting. 
The conditional target \eqref{eq:loso-target} is in any case unavailable here, since 204 of the 429 families are singletons with no clustermates to condition on.
The encompassing model converged to a proper solution with 132 free parameters, and the selected structure retained $\eta\to E$ and $\eta\to O$, but omitted $E\to\eta$ and $O\to\eta$ (Figure \ref{fig:midus-search}a).
Relative to it, adding $O\to\eta$ was worse by 13.8 elpd units (paired SE 2.7), and adding both trait-leading paths, which recovers the encompassing model, worse by 27.0 (SE 3.6).
In a separate fit of the selected model the shared unstandardised coefficients were $\widehat\beta_{E\eta}=0.164$ (posterior SD 0.035) and $\widehat\beta_{O\eta}=0.146$ (SD 0.033).
Prediction for a new family therefore supported PWB-leading rather than reciprocal prospective connectivity for both Plasticity traits.

\begin{figure}[htbp]
\centering
\includegraphics[width=\linewidth]{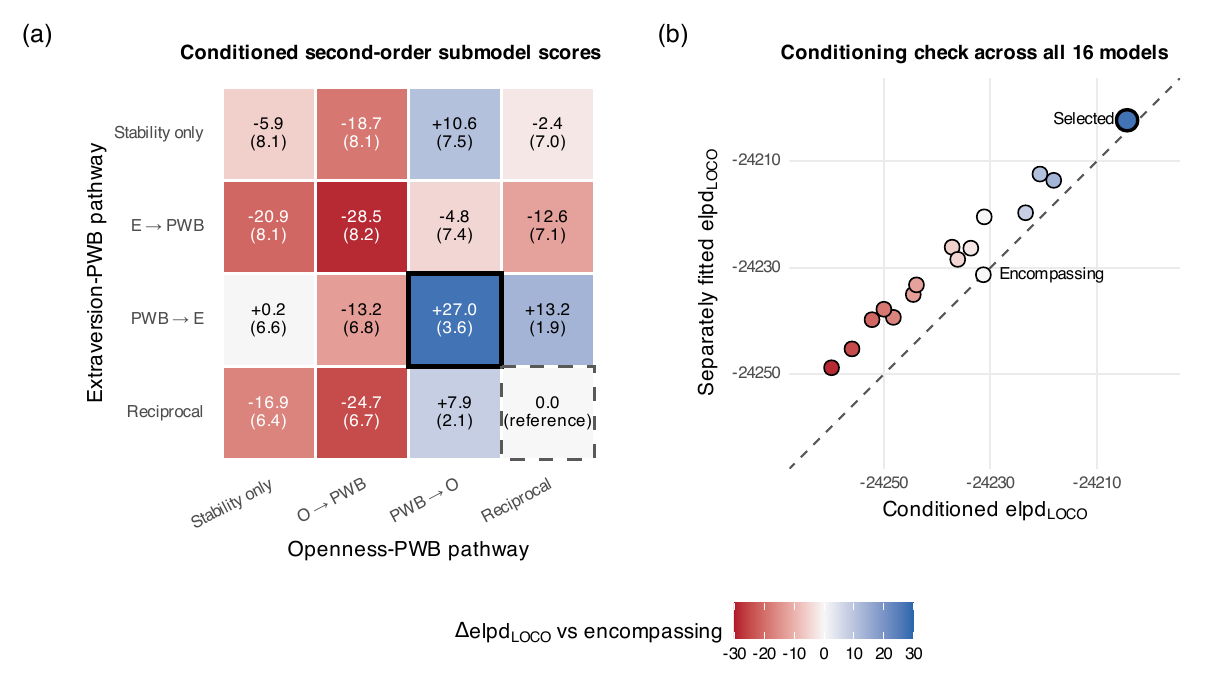}
\caption{The 16 model structures, scored by family-\loco{}.
(a) Conditioned second-order Taylor $\Delta\elpd_{\loco}$ relative to the encompassing model and SE in parentheses, by Extraversion--PWB (rows) and Openness--PWB (columns) pathways; positive values favour the restriction, with the selected structure outlined.
(b) The same scores against their separately fitted counterparts, with the equality line dashed; every conditioned value comes from a single encompassing fit.}
\label{fig:midus-search}
\end{figure}

Separate fits leave that selection unchanged, with rank agreement over all 16 structures at Kendall's $\tau=0.88$.
Conditioned scores were uniformly conservative by an amount that tracked the joint distance \eqref{eq:cond-distance} of the restriction from the encompassing centre. 
The three candidates at $D_m \le 2.11$ were under-scored by 2.0 to 4.5 elpd units, the twelve at $D_m \ge 4.40$ by 7.3 to 12.5---the far-from-centre degradation of Section~\ref{sec:sim-conditioning}, which only affects structures the data already reject, pushing them further out of contention while leaving the decision intact.

A structural hypothesis space can therefore be enumerated by this conditioning strategy rather than narrowed by hand, with every candidate scored on the target the design calls for.
In this case, enumerating all 16 structures revises an earlier finding on these waves. 
\citeauthor{joshanloo2023reciprocal}'s (\citeyear{joshanloo2023reciprocal}) within-person model found psychological well-being and both Plasticity traits to be reciprocally related.
Scored out of family, the trait-leading half of that reciprocity does not pay for itself, and dropping it improves elpd by 27.0 units, so that the leading process for both traits is well-being.
That is a predictive ordering and not a causal one, drawn on a family-relative rather than person-relative decomposition and with E and O observed rather than latent.
It is nonetheless a distinction this criterion can draw, and one that the cost of refitting has until now kept out of reach.

\section{Discussion}
\label{sec:discussion}

The marginal, cluster-level predictive density is the correct target for multilevel SEM whenever a structural theory is meant to generalise to new clusters \parencite{merkle2019bayesian}, but what has been lacking is an affordable way to compute it, since the existing routes either refit the model once per deleted cluster or reweight a posterior sample by importance weights that fail exactly where a held-out cluster is influential enough to matter.
Our work takes
a single \texttt{INLAvaan} fit and returns the \loco{} predictive score for the encompassing model and for every submodel obtained by zeroing structural paths, with no refits, no Monte Carlo over the latent state, and no harmonic-mean variance pathology.
The \loco~cost is one Gaussian conditional update and one $O(Jd^2)$ post-processing pass per scored submodel, small relative to the encompassing fit itself.

The construction replaces one kind of numerical error with another, without changing what is estimated.
Since $\log p(\mathbf{y}_j \mid \mathbf{y}_{-j})$ is a fixed functional of the posterior, a Monte Carlo evaluation returns a random value whose error is sampling noise, while the Taylor score returns a deterministic value whose error is that of the truncation itself.
That error depends on how far deleting cluster $j$ displaces the posterior, hence on whether any one cluster carries an appreciable share of the information in the sample.
In simulation, the truncation error of the second-order Taylor score never exceeds $0.061$ elpd per cluster, and in the two real-data applications the score agrees with MCMC$+$PSIS wherever PSIS is itself reliable and reproduces the selection that separate refits deliver across an enumerated candidate space.

The reach of the construction is delimited by its assumptions.
The closed-form cluster likelihood \eqref{eq:cluster-loglik} requires continuous indicators and Gaussian latents.
The cluster-deleted posterior identity \eqref{eq:cluster-deleted-id} requires only cluster independence given $\bm\vartheta$, and so holds in much broader settings.
The Taylor score of Section~\ref{sec:laplace} requires only a twice-differentiable per-unit log-likelihood and a Gaussian posterior summary $(\bm\vartheta^\ast, \bm\Omega)$, so it generalises to any model for which such a summary can be formed.
In particular, the summary need not come from a Laplace-based optimisation at all.
The mean and covariance of a set of MCMC draws estimate exactly such a summary, and a centre near the posterior mean is in fact preferable to the mode \autocite{vanniekerk2024lowrank}, so moment matching the draws and applying \eqref{eq:laplace-1st}--\eqref{eq:laplace-2nd} verbatim turns the posterior sample of any Bayesian SEM engine into the same Taylor elpd scores, submodel conditioning and existence check included.
For MCMC users this is an alternative route to the cluster-level score that may strain under importance sampling.

Comparing the first- and second-order Taylor scores \eqref{eq:laplace-1st} and \eqref{eq:laplace-2nd} provides a cheap internal check at no additional cost, their gap being the curvature correction the expansion makes for each cluster.
These gaps are small where the truncation has settled, large where it has not, and sum over clusters to $\tfrac12 p_{\mathrm{eff}}$, so a total materially exceeding $\tfrac12 \sum_j \operatorname{tr}(-\bm\Omega \mathbf{H}_j)$, a quantity the same post-processing pass already returns, signals that the expansion has not converged for the sample as a whole.
The existence condition \eqref{eq:existence-2nd} is a second such check.
Its failure indicates breakdown of the local quadratic approximation for that deletion, rather than failure of the exact \loco{} predictive density.
Since $\mathbf{A}_j$ is the approximate cluster-deleted precision, $\mathbf{A}_j \not\succ 0$ means that the deleted summary is improper, the curvature attributed to cluster $j$ meeting or exceeding the entire posterior precision in some direction of the parameter space.

The eigenvalue $k_j$ of Section~\ref{sec:laplace} measures this leverage, which a well-identified model keeps small.
Multilevel SEM are models for fitting two covariance matrices, $\bm\Sigma_{\mathrm{w}}$ to the pooled within-cluster scatter and $\bm\Sigma_{\mathrm{b}}$ to the spread of the cluster means.
Each cluster contributes its scatter $\mathbf{S}_j$ to the former and its mean $\bar{\mathbf{y}}_j$ to the latter, so no parameter element in $\bm\vartheta$ belongs to any one cluster, and a cluster's share of the evidence is roughly $n_j/N$ for $\bm\Sigma_{\mathrm{w}}$ and $1/J$ for $\bm\Sigma_{\mathrm{b}}$.
A single cluster can therefore dominate a direction only by being very large, by lying far from the other cluster means, or by being the only cluster whose data carry information about some parameter at all.
If a violation happens to occur, it is then a prompt to inspect the data (why does cluster $j$ contribute so much curvature?) or the specification (is some parameter identified by cluster $j$ alone, such as a path from a covariate that varies in one cluster only?).

Methodologically, the submodel scoring of Section~\ref{sec:restrictions} is a fast surrogate for projection-predictive variable selection \parencite{piironen2020projective}, but defined directly through the predictive score rather than through a projection onto a reference model.
The rank-one update \eqref{eq:restr-rank1} plays a role analogous to projecting onto a constrained submodel, with closed-form Gaussian conditioning in place of a per-candidate Kullback--Leibler optimisation, and no input beyond the summaries the encompassing fit already produces.
The conditioning step of Section~\ref{sec:restr-laplace} is exact under the Laplace summary of the encompassing posterior and so inherits only that summary's error.
In directions where the posterior is approximately Gaussian the restricted summary is essentially a refit, as borne out by the agreement reported above, whereas in directions where it is highly skewed or multimodal the procedure is best treated as a fast screen, with selected submodels ideally validated by full refits.
Because the conditioning error of Section~\ref{sec:sim-conditioning} is one-sided, that screen errs in one direction only, wrongly excluding a candidate whose restriction sits far from the centre but never spuriously promoting one, and $D_m$ \eqref{eq:cond-distance} names the candidates at risk before any scoring is done.

Three future directions are discussed.
First, genuinely non-nested or differently parameterised alternatives, including candidates that change the prior and so leave the compatible-prior convention of Section~\ref{sec:restrictions}, lie outside the conditioning step and require a second fit.
But that fit costs seconds, its Taylor scores target the same estimand as the first, and the two candidates would compare directly.
Second, the cluster log-likelihood is the only ingredient that changes across designs, so deeper nesting, multigroup data, and incomplete data all enter through the kernel alone, with the deletion identity, the Taylor score, and the conditioning step untouched.
Measurement-invariance constraints \parencite{meredith1993measurement} are linear restrictions on a group-indexed $\bm\vartheta$, so configural, metric, and scalar specifications could in principle be scored by conditioning a single configural fit.
With incomplete data, a full-information kernel \parencite{enders2022applied} scores each unit on the entries it has, valid under missingness at random \parencite{rubin1976inference} provided candidates share the same observed entries.
Third, and most interestingly, non-Gaussian indicators preclude the closed form \eqref{eq:cluster-loglik}, since the latent state cannot be integrated out analytically as in \textcite{liu2025leavegroupout}.
Yet, the rest of the methodology translates seamlessly. 
By simply approximating this outer integral---computing $\ell_j$, $\mathbf{s}_j$, and $\mathbf{H}_j$ through numerical or Laplace integration---the entirety of the downstream pipeline applies exactly as before.

\section*{CRediT Author Statement}

Mohammad Alhyari: Conceptualization, Formal analysis, Investigation, Methodology, Project administration, Writing - Original Draft; Haziq Jamil: Data curation, Formal analysis, Funding acquisition, Investigation, Software, Validation, Visualization, Writing - Original Draft, Writing - Review \& Editing; Hans Montcho: Conceptualization, Writing - Review \& Editing; Håvard Rue: Funding acquisition, Supervision, Validation, Writing - Review \& Editing.

\section*{Acknowledgements}

This publication is based upon work supported by the King Abdullah University of Science and Technology (KAUST) Research Funding Office under Award No. URF/1/6921-01-01.

\section*{Data Availability}

The two empirical datasets used in this article are publicly accessible.
The MIDUS data \autocite{radler2014midlife} are archived in the Midlife in the United States (MIDUS) Series at the Inter-university Consortium for Political and Social Research (ICPSR), and the three waves analysed here are distributed as ICPSR 2760 (MIDUS~1, 1995--1996), ICPSR 4652 (MIDUS~2, 2004--2006), and ICPSR 36346 (MIDUS~3, 2013--2014), available at \url{https://www.icpsr.umich.edu/web/ICPSR/series/203} subject to free registration.
The PISA 2022 data are available from the OECD PISA 2022 Database \autocite{oecd2024pisa}.
All code required to fully reproduce the results is deposited in an open repository at \url{https://osf.io/7au5h/overview?view_only=965ceb3d57c34fbd96689407373f0ea5}, and reported timings are wall-clock on an Apple M4 Pro (14 cores, 24\,GB) under R 4.5.2, with the closed-form scoring passes parallelised over 13 cores.

\printbibliography

@book{bollen1989structural,
  title = {Structural Equations with Latent Variables},
  author = {Bollen, Kenneth A.},
  year = 1989,
  series = {Wiley {{Series}} in {{Probability}} and {{Statistics}}},
  edition = {1},
  publisher = {Wiley},
  doi = {10.1002/9781118619179},
  urldate = {2026-08-31},
  isbn = {978-0-471-01171-2},
  langid = {english}
}

@article{consonni2008compatibility,
  title = {Compatibility of Prior Specifications across Linear Models},
  author = {Consonni, Guido and Veronese, Piero},
  year = 2008,
  journal = {Statistical Science},
  volume = {23},
  number = {3},
  pages = {332--353},
  publisher = {Institute of Mathematical Statistics},
  issn = {0883-4237},
  doi = {10.1214/08-STS258}
}

@article{deyoung2015cybernetic,
  title = {Cybernetic {{Big Five}} Theory},
  author = {DeYoung, Colin G.},
  year = 2015,
  journal = {Journal of Research in Personality},
  volume = {56},
  pages = {33--58},
  issn = {00926566},
  doi = {10.1016/j.jrp.2014.07.004},
  urldate = {2026-05-30},
  langid = {english}
}

@article{dickey1971weighted,
  title = {The Weighted Likelihood Ratio, Linear Hypotheses on Normal Location Parameters},
  author = {Dickey, James M.},
  year = 1971,
  journal = {The Annals of Mathematical Statistics},
  volume = {42},
  number = {1},
  pages = {204--223},
  publisher = {Institute of Mathematical Statistics},
  issn = {0003-4851},
  doi = {10.1214/aoms/1177693507}
}

@article{digman1997higherorder,
  title = {Higher-Order Factors of the {{Big Five}}},
  author = {Digman, John M.},
  year = 1997,
  journal = {Journal of Personality and Social Psychology},
  volume = {73},
  number = {6},
  pages = {1246--1256},
  issn = {1939-1315, 0022-3514},
  doi = {10.1037/0022-3514.73.6.1246},
  urldate = {2026-05-30},
  langid = {english}
}

@book{enders2022applied,
  title = {Applied Missing Data Analysis},
  author = {Enders, Craig K.},
  year = 2022,
  edition = {2},
  publisher = {Guilford Press},
  urldate = {2026-03-22}
}

@article{freni2025graphical,
  title = {A Graphical Framework for Interpretable Correlation Matrix Models for Multivariate Regression},
  author = {{Freni-Sterrantino}, Anna and Rustand, Denis and {van Niekerk}, Janet and Krainski, Elias and Rue, H{\aa}vard},
  year = 2025,
  journal = {Statistical Methods \& Applications},
  volume = {34},
  number = {3},
  pages = {409--447},
  issn = {1618-2510, 1613-981X},
  doi = {10.1007/s10260-025-00788-y},
  urldate = {2025-08-18},
  langid = {english}
}

@misc{freni2026informative,
  title = {Informative Distance-Based Priors for Correlation Matrices Centred on a Target Reference},
  author = {{Freni-Sterrantino}, Anna and {van Niekerk}, Janet and Krainski, Elias Teixeira and Rustand, Denis and Jamil, Haziq and Rue, H{\aa}vard},
  year = 2026,
  number = {arXiv:2607.18926},
  eprint = {2607.18926},
  primaryclass = {stat.ME},
  publisher = {arXiv},
  doi = {10.48550/arXiv.2607.18926},
  urldate = {2026-08-19},
  archiveprefix = {arXiv},
  organization = {arXiv}
}

@incollection{gelfand1992model,
  title = {Model Determination Using Predictive Distributions with Implementation via Sampling-Based Methods},
  booktitle = {Bayesian Statistics 4},
  author = {Gelfand, A. E. and Dey, D. K. and Chang, H.},
  editor = {Bernardo, J. M. and Berger, J. O. and Dawid, A. P. and Smith, A. F. M.},
  year = 1992,
  pages = {147--167},
  publisher = {Oxford University Press},
  doi = {10.1093/oso/9780198522669.003.0009}
}

@book{gelman2026bayesian,
  title = {Bayesian Workflow},
  author = {Gelman, Andrew and Vehtari, Aki and McElreath, Richard and Simpson, Daniel and Margossian, Charles C. and Yao, Yuling and Kennedy, Lauren and Gabry, Jonah and B{\"u}rkner, Paul-Christian and Modr{\'a}k, Martin and Barajas, Vianey Leos},
  year = 2026,
  publisher = {Chapman \& Hall/CRC},
  doi = {10.1201/9781003044024}
}

@article{hsu2015detecting,
  title = {Detecting Misspecified Multilevel Structural Equation Models with Common Fit Indices: {{A Monte Carlo}} Study},
  author = {Hsu, Hsien-Yuan and Kwok, Oi-Man and Lin, Jr-Huang and Acosta, Sandra},
  year = 2015,
  journal = {Multivariate Behavioral Research},
  volume = {50},
  number = {2},
  pages = {197--215},
  publisher = {Informa UK Limited},
  issn = {0027-3171},
  doi = {10.1080/00273171.2014.977429}
}

@article{jacobucci2016regularized,
  title = {Regularized Structural Equation Modeling},
  author = {Jacobucci, Ross and Grimm, Kevin J. and McArdle, John J.},
  year = 2016,
  journal = {Structural Equation Modeling: A Multidisciplinary Journal},
  volume = {23},
  number = {4},
  pages = {555--566},
  issn = {1070-5511, 1532-8007},
  doi = {10.1080/10705511.2016.1154793},
  urldate = {2025-05-16},
  langid = {english}
}

@misc{jamil2026approximate,
  title = {Approximate {{Bayesian}} Inference for Structural Equation Models Using Integrated Nested {{Laplace}} Approximations},
  author = {Jamil, Haziq and Rue, H{\aa}vard},
  year = 2026,
  number = {2603.25690 [stat.ME]},
  eprint = {2603.25690},
  primaryclass = {stat.ME},
  publisher = {arXiv},
  doi = {10.48550/arXiv.2603.25690},
  archiveprefix = {arXiv},
  copyright = {Creative Commons Attribution Non Commercial Share Alike 4.0 International},
  organization = {arXiv}
}

@misc{jamil2026implementation,
  title = {Implementation and Workflows for {{INLA-based}} Approximate {{Bayesian}} Structural Equation Modelling},
  author = {Jamil, Haziq and Rue, H{\aa}vard},
  year = 2026,
  number = {2604.00671 [stat.CO]},
  eprint = {2604.00671},
  primaryclass = {stat.CO},
  publisher = {arXiv},
  doi = {10.48550/arXiv.2604.00671},
  archiveprefix = {arXiv},
  copyright = {Creative Commons Attribution Non Commercial Share Alike 4.0 International},
  langid = {english},
  organization = {arXiv}
}

@article{joshanloo2023reciprocal,
  title = {Reciprocal Relationships between Personality Traits and Psychological Well-Being},
  author = {Joshanloo, Mohsen},
  year = 2023,
  journal = {British Journal of Psychology},
  volume = {114},
  number = {1},
  pages = {54--69},
  publisher = {Wiley},
  issn = {0007-1269},
  doi = {10.1111/bjop.12596}
}

@article{klugkist2007bayes,
  title = {The {{Bayes}} Factor for Inequality and about Equality Constrained Models},
  author = {Klugkist, Irene and Hoijtink, Herbert},
  year = 2007,
  journal = {Computational Statistics \& Data Analysis},
  volume = {51},
  number = {12},
  pages = {6367--6379},
  publisher = {Elsevier BV},
  issn = {0167-9473},
  doi = {10.1016/j.csda.2007.01.024}
}

@report{lachman1997midlife,
  title = {The {{Midlife Development Inventory}} ({{MIDI}}) Personality Scales: {{Scale}} Construction and Scoring},
  author = {Lachman, M. E. and Weaver, S. L.},
  year = 1997,
  institution = {Brandeis University},
  url = {https://www.brandeis.edu/psychology/lachman/pdfs/midi-personality-scales.pdf}
}

@article{liu2025leavegroupout,
  title = {Leave-Group-out Cross-Validation for Latent {{Gaussian}} Models},
  author = {Liu, Zhedong and {van Niekerk}, Janet and Rue, H{\aa}vard},
  year = 2025,
  journal = {SORT-Statistics and Operations Research Transactions},
  volume = {49},
  number = {1},
  pages = {121--146},
  publisher = {Institut d'Estad\'istica de Catalunya},
  address = {ES},
  issn = {1696-2281, 2013-8830},
  doi = {10.57645/20.8080.02.25},
  urldate = {2026-05-03},
  langid = {english}
}

@article{lu2016bayesian,
  title = {Bayesian Factor Analysis as a Variable-Selection Problem: {{Alternative}} Priors and Consequences},
  author = {Lu, Zhao-Hua and Chow, Sy-Miin and Loken, Eric},
  year = 2016,
  journal = {Multivariate Behavioral Research},
  volume = {51},
  number = {4},
  pages = {519--539},
  doi = {10.1080/00273171.2016.1168279}
}

@article{meredith1993measurement,
  title = {Measurement Invariance, Factor Analysis and Factorial Invariance},
  author = {Meredith, William},
  year = 1993,
  journal = {Psychometrika},
  volume = {58},
  number = {4},
  pages = {525--543},
  issn = {1860-0980},
  doi = {10.1007/BF02294825},
  urldate = {2025-09-09},
  langid = {english}
}

@article{merkle2018blavaan,
  title = {{{blavaan}}: {{Bayesian}} Structural Equation Models via Parameter Expansion},
  shorttitle = {Blavaan},
  author = {Merkle, Edgar C. and Rosseel, Yves},
  year = 2018,
  journal = {Journal of Statistical Software},
  volume = {85},
  number = {4},
  pages = {1--30},
  issn = {1548-7660},
  doi = {10.18637/jss.v085.i04},
  urldate = {2023-11-28},
  copyright = {Copyright (c) 2018 Edgar C. Merkle, Yves Rosseel},
  langid = {english}
}

@article{merkle2019bayesian,
  title = {Bayesian Comparison of Latent Variable Models: {{Conditional}} versus Marginal Likelihoods},
  shorttitle = {Bayesian Comparison of Latent Variable Models},
  author = {Merkle, Edgar C. and Furr, Daniel and {Rabe-Hesketh}, Sophia},
  year = 2019,
  journal = {Psychometrika},
  volume = {84},
  number = {3},
  pages = {802--829},
  issn = {0033-3123, 1860-0980},
  doi = {10.1007/s11336-019-09679-0},
  urldate = {2025-12-23},
  langid = {english}
}

@article{merkle2021efficient,
  title = {Efficient {{Bayesian}} Structural Equation Modeling in {{Stan}}},
  author = {Merkle, Edgar C. and Fitzsimmons, Ellen and Uanhoro, James and Goodrich, Ben},
  year = 2021,
  journal = {Journal of Statistical Software},
  volume = {100},
  number = {6},
  pages = {1--22},
  issn = {1548-7660},
  doi = {10.18637/jss.v100.i06},
  urldate = {2024-05-12},
  copyright = {Copyright (c) 2021 Edgar C. Merkle, Ellen Fitzsimmons, James Uanhoro, Ben Goodrich},
  langid = {english}
}

@article{merkle2026nuances,
  title = {Nuances of Information Criteria for {{Bayesian}} Psychometric Models},
  author = {Merkle, Edgar C.},
  year = 2026,
  journal = {Methodology},
  volume = {22},
  number = {3},
  pages = {195--224},
  issn = {1614-2241},
  doi = {10.5964/meth.20361},
  urldate = {2026-08-31}
}

@article{muthen1994multilevel,
  title = {Multilevel Covariance Structure Analysis},
  author = {Muth{\'e}n, Bengt},
  year = 1994,
  journal = {Sociological Methods \& Research},
  volume = {22},
  number = {3},
  pages = {376--398},
  issn = {0049-1241, 1552-8294},
  doi = {10.1177/0049124194022003006},
  urldate = {2025-09-02},
  copyright = {https://journals.sagepub.com/page/policies/text-and-data-mining-license},
  langid = {english}
}

@article{muthen2012bayesian,
  title = {Bayesian Structural Equation Modeling: {{A}} More Flexible Representation of Substantive Theory},
  shorttitle = {Bayesian Structural Equation Modeling},
  author = {Muth{\'e}n, Bengt and Asparouhov, Tihomir},
  year = 2012,
  journal = {Psychological Methods},
  volume = {17},
  number = {3},
  pages = {313--335},
  publisher = {American Psychological Association},
  address = {US},
  issn = {1939-1463},
  doi = {10.1037/a0026802}
}

@book{oecd2024pisa,
  title = {{{PISA}} 2022 Technical Report},
  author = {{OECD}},
  year = 2024,
  publisher = {OECD Publishing},
  address = {Paris},
  doi = {10.1787/01820d6d-en},
  urldate = {2024-09-18},
  langid = {english}
}

@article{penny2013efficient,
  title = {Efficient Posterior Probability Mapping Using {{Savage}}--{{Dickey}} Ratios},
  author = {Penny, William D. and Ridgway, Gerard R.},
  year = 2013,
  journal = {PLOS ONE},
  volume = {8},
  number = {3},
  pages = {e59655},
  publisher = {Public Library of Science (PLoS)},
  issn = {1932-6203},
  doi = {10.1371/journal.pone.0059655}
}

@article{piironen2020projective,
  title = {Projective Inference in High-Dimensional Problems: {{Prediction}} and Feature Selection},
  shorttitle = {Projective Inference in High-Dimensional Problems},
  author = {Piironen, Juho and Paasiniemi, Markus and Vehtari, Aki},
  year = 2020,
  journal = {Electronic Journal of Statistics},
  volume = {14},
  number = {1},
  pages = {2155--2197},
  publisher = {{Institute of Mathematical Statistics and Bernoulli Society}},
  issn = {1935-7524, 1935-7524},
  doi = {10.1214/20-EJS1711},
  urldate = {2026-05-19},
  langid = {english}
}

@article{radler2014midlife,
  title = {The {{Midlife}} in the {{United States}} ({{MIDUS}}) Series: {{A}} National Longitudinal Study of Health and Well-Being},
  author = {Radler, Barry T.},
  year = 2014,
  journal = {Open Health Data},
  volume = {2},
  pages = {e3},
  publisher = {Ubiquity Press, Ltd.},
  issn = {2054-7102},
  doi = {10.5334/ohd.ai}
}

@article{rosseel2012lavaan,
  title = {{{lavaan}}: {{An R}} Package for Structural Equation Modeling},
  shorttitle = {{\textbf{Lavaan}}},
  author = {Rosseel, Yves},
  year = 2012,
  journal = {Journal of Statistical Software},
  volume = {48},
  number = {2},
  pages = {1--36},
  issn = {1548-7660},
  doi = {10.18637/jss.v048.i02},
  urldate = {2023-11-23},
  langid = {english}
}

@article{rosseel2021evaluating,
  title = {Evaluating the Observed Log-Likelihood Function in Two-Level Structural Equation Modeling with Missing Data: {{From}} Formulas to {{R}} Code},
  shorttitle = {Evaluating the Observed Log-Likelihood Function in Two-Level Structural Equation Modeling with Missing Data},
  author = {Rosseel, Yves},
  year = 2021,
  journal = {Psych},
  volume = {3},
  number = {2},
  pages = {197--232},
  issn = {2624-8611},
  doi = {10.3390/psych3020017},
  urldate = {2025-09-02},
  copyright = {https://creativecommons.org/licenses/by/4.0/},
  langid = {english}
}

@article{rubin1976inference,
  title = {Inference and Missing Data},
  author = {Rubin, Donald B.},
  year = 1976,
  journal = {Biometrika},
  volume = {63},
  number = {3},
  pages = {581--592},
  issn = {0006-3444},
  doi = {10.1093/biomet/63.3.581},
  urldate = {2024-02-26}
}

@article{rue2009approximate,
  title = {Approximate {{Bayesian}} Inference for Latent {{Gaussian}} Models by Using Integrated Nested {{Laplace}} Approximations},
  author = {Rue, H{\aa}vard and Martino, Sara and Chopin, Nicolas},
  year = 2009,
  journal = {Journal of the Royal Statistical Society: Series B (Statistical Methodology)},
  volume = {71},
  number = {2},
  pages = {319--392},
  issn = {1369-7412, 1467-9868},
  doi = {10.1111/j.1467-9868.2008.00700.x},
  urldate = {2023-11-27},
  langid = {english}
}

@article{ryff1995structure,
  title = {The Structure of Psychological Well-Being Revisited},
  author = {Ryff, Carol D. and Keyes, Corey Lee M.},
  year = 1995,
  journal = {Journal of Personality and Social Psychology},
  volume = {69},
  number = {4},
  pages = {719--727},
  publisher = {American Psychological Association (APA)},
  issn = {1939-1315},
  doi = {10.1037/0022-3514.69.4.719}
}

@article{ryu2009levelspecific,
  title = {Level-Specific Evaluation of Model Fit in Multilevel Structural Equation Modeling},
  author = {Ryu, Ehri and West, Stephen G.},
  year = 2009,
  journal = {Structural Equation Modeling: A Multidisciplinary Journal},
  volume = {16},
  number = {4},
  pages = {583--601},
  issn = {1070-5511, 1532-8007},
  doi = {10.1080/10705510903203466},
  urldate = {2025-09-02},
  langid = {english}
}

@article{sivula2025uncertainty,
  title = {Uncertainty in {{Bayesian Leave-One-Out Cross-Validation Based Model Comparison}}},
  author = {Sivula, Tuomas and Magnusson, M{\aa}ns and Matamoros, Asael Alonzo and Vehtari, Aki},
  year = 2025,
  journal = {Bayesian Analysis},
  issn = {1936-0975},
  doi = {10.1214/25-BA1569},
  urldate = {2026-09-01}
}

@article{spiegelhalter2002bayesian,
  title = {Bayesian Measures of Model Complexity and Fit},
  author = {Spiegelhalter, David J. and Best, Nicola G. and Carlin, Bradley P. and Van Der Linde, Angelika},
  year = 2002,
  journal = {Journal of the Royal Statistical Society Series B: Statistical Methodology},
  volume = {64},
  number = {4},
  pages = {583--639},
  issn = {1369-7412, 1467-9868},
  doi = {10.1111/1467-9868.00353},
  urldate = {2026-09-01},
  copyright = {https://academic.oup.com/journals/pages/open\_access/funder\_policies/chorus/standard\_publication\_model},
  langid = {english}
}

@article{tierney1986accurate,
  title = {Accurate Approximations for Posterior Moments and Marginal Densities},
  author = {Tierney, Luke and Kadane, Joseph B.},
  year = 1986,
  journal = {Journal of the American Statistical Association},
  volume = {81},
  number = {393},
  pages = {82--86},
  issn = {0162-1459, 1537-274X},
  doi = {10.1080/01621459.1986.10478240},
  urldate = {2026-02-24},
  langid = {english}
}

@article{vanniekerk2023new,
  title = {A New Avenue for {{Bayesian}} Inference with {{INLA}}},
  author = {{van Niekerk}, Janet and Krainski, Elias and Rustand, Denis and Rue, H{\aa}vard},
  year = 2023,
  journal = {Computational Statistics \& Data Analysis},
  volume = {181},
  pages = {107692},
  issn = {01679473},
  doi = {10.1016/j.csda.2023.107692},
  urldate = {2023-11-22},
  langid = {english}
}

@article{vanniekerk2024lowrank,
  title = {Low-Rank Variational {{Bayes}} Correction to the {{Laplace}} Method},
  author = {{van Niekerk}, Janet and Rue, H{\aa}vard},
  year = 2024,
  journal = {Journal of Machine Learning Research},
  volume = {25},
  number = {62},
  pages = {1--25},
  langid = {english}
}

@article{vehtari2017practical,
  title = {Practical {{Bayesian}} Model Evaluation Using Leave-One-out Cross-Validation and {{WAIC}}},
  author = {Vehtari, Aki and Gelman, Andrew and Gabry, Jonah},
  year = 2017,
  journal = {Statistics and Computing},
  volume = {27},
  number = {5},
  pages = {1413--1432},
  issn = {1573-1375},
  doi = {10.1007/s11222-016-9696-4},
  urldate = {2025-12-23},
  langid = {english}
}

@article{vehtari2024pareto,
  title = {Pareto Smoothed Importance Sampling},
  author = {Vehtari, Aki and Simpson, Daniel and Gelman, Andrew and Yao, Yuling and Gabry, Jonah},
  year = 2024,
  journal = {Journal of Machine Learning Research},
  volume = {25},
  number = {72},
  pages = {1--58},
  url = {https://jmlr.org/papers/v25/19-556.html}
}

@article{verdinelli1995computing,
  title = {Computing {{Bayes}} Factors Using a Generalization of the {{Savage}}--{{Dickey}} Density Ratio},
  author = {Verdinelli, Isabella and Wasserman, Larry},
  year = 1995,
  journal = {Journal of the American Statistical Association},
  volume = {90},
  number = {430},
  pages = {614--618},
  publisher = {Informa UK Limited},
  issn = {0162-1459},
  doi = {10.1080/01621459.1995.10476554}
}

@article{watanabe2010asymptotic,
  title = {Asymptotic Equivalence of {{Bayes}} Cross Validation and Widely Applicable Information Criterion in Singular Learning Theory},
  author = {Watanabe, Sumio},
  year = 2010,
  journal = {Journal of Machine Learning Research},
  volume = {11},
  number = {116},
  pages = {3571--3594},
  url = {https://www.jmlr.org/papers/v11/watanabe10a.html}
}

@article{wolpert2012alpha,
  title = {{$\alpha$}-Stable Limit Laws for Harmonic Mean Estimators of Marginal Likelihoods},
  author = {Wolpert, Robert L. and Schmidler, Scott C.},
  year = 2012,
  journal = {Statistica Sinica},
  volume = {22},
  number = {3},
  pages = {1233--1251},
  publisher = {Statistica Sinica (Institute of Statistical Science)},
  issn = {1017-0405},
  doi = {10.5705/ss.2010.221}
}

\appendix

\section{Joint and Conditional Treatment of Exogenous Covariates}
\label{app:covariates}

The main text models exogenous covariates jointly with the indicators.
Software offers a second convention, and this appendix details the practical consequences of each choice.
Under the joint formulation (\texttt{fixed.x = FALSE}) the covariates receive a saturated Gaussian block and are modelled alongside the indicators; under the conditional formulation (\texttt{fixed.x = TRUE}, the \texttt{lavaan} default) their moments are held at the sample values $(\bar{\mathbf{x}}, \mathbf{S}_x)$ and only the outcomes are modelled.
The saturated block is variation-independent of the structural parameters, so the two return identical structural estimates, and Figure~\ref{fig-fixedx} shows them side by side.

\begin{figure}[htb]
\centering
\begin{subfigure}[b]{0.47\linewidth}
\centering
\begin{tikzpicture}[>=stealth]
\node[lat] (eta) at (0,0)       {$\eta$};
\node[obs] (x)   at (-2.1,0)    {$x$};
\node[obs] (y1)  at (-1.2,-1.7) {$y_1$};
\node[obs] (y2)  at (0,-1.7)    {$y_2$};
\node[obs] (y3)  at (1.2,-1.7)  {$y_3$};

\draw[->] (x) -- node[lbl, above] {$\gamma$} (eta);
\draw[->] (eta) -- node[lbl, left, pos=0.6] {$\lambda_1$} (y1);
\draw[->] (eta) -- node[lbl, right, pos=0.6] {$\lambda_2$} (y2);
\draw[->] (eta) -- node[lbl, right, pos=0.6] {$\lambda_3$} (y3);

\draw[<->] ($(x.north)+(-1.6mm,0)$) to[out=115,in=65,looseness=5]
  node[lbl, above] {$\phi$} ($(x.north)+(1.6mm,0)$);

\draw[<-] (eta) -- ++(0.95,0.5) node[lbl, right] {$\psi$};
\draw[<-] (y1) -- ++(0,-0.8);
\draw[<-] (y2) -- ++(0,-0.8);
\draw[<-] (y3) -- ++(0,-0.8);
\end{tikzpicture}
\vspace{2mm}
\caption{Joint (\texttt{fixed.x = FALSE}):\\ scores $\log p(y_j, \mathbf{x}_j \mid \mathcal{D}_{-j})$.}
\end{subfigure}
\hfill
\begin{subfigure}[b]{0.47\linewidth}
\centering
\begin{tikzpicture}[>=stealth]
\node[lat]   (eta) at (0,0)       {$\eta$};
\node[given] (x)   at (-2.1,0)    {$x$};
\node[obs]   (y1)  at (-1.2,-1.7) {$y_1$};
\node[obs]   (y2)  at (0,-1.7)    {$y_2$};
\node[obs]   (y3)  at (1.2,-1.7)  {$y_3$};

\draw[->] (x) -- node[lbl, above] {$\gamma$} (eta);
\draw[->] (eta) -- node[lbl, left, pos=0.6] {$\lambda_1$} (y1);
\draw[->] (eta) -- node[lbl, right, pos=0.6] {$\lambda_2$} (y2);
\draw[->] (eta) -- node[lbl, right, pos=0.6] {$\lambda_3$} (y3);

\node[lbl, above=2.5pt of x] {given};

\draw[<-] (eta) -- ++(0.95,0.5) node[lbl, right] {$\psi$};
\draw[<-] (y1) -- ++(0,-0.8);
\draw[<-] (y2) -- ++(0,-0.8);
\draw[<-] (y3) -- ++(0,-0.8);
\end{tikzpicture}
\vspace{2mm}
\caption{Conditional (\texttt{fixed.x = TRUE}):\\ scores $\log p(y_j \mid \mathbf{x}_j, \mathcal{D}_{-j})$.}
\end{subfigure}
\captionsetup{justification=justified}
\caption{The two covariate formulations, for a toy model with one covariate, one latent construct, and three indicators. Shading marks the variables to which the model assigns a distribution. Every parameter is common to both panels and estimated identically. The formulations differ only in whether the covariate carries parameters of its own, here the variance $\phi$.}
\label{fig-fixedx}
\end{figure}

The two part company in what they predict.
Writing $\mathbf{x}_j$ for the held-out cluster's covariates and $\mathcal{D}_{-j}$ for all data with cluster $j$ deleted in its entirety, the leave-one-cluster-out estimands are
\begin{equation*}
\elpd^{\mathrm{J}} = \sum_{j=1}^J \log p(\mathbf{y}_j, \mathbf{x}_j \mid \mathcal{D}_{-j}),
\qquad
\elpd^{\mathrm{C}} = \sum_{j=1}^J \log p(\mathbf{y}_j \mid \mathbf{x}_j, \mathcal{D}_{-j}),
\end{equation*}
linked pointwise by the predictive identity
\begin{equation}
\log p(\mathbf{y}_j, \mathbf{x}_j \mid \mathcal{D}_{-j}) = \log p(\mathbf{y}_j \mid \mathbf{x}_j, \mathcal{D}_{-j}) + \log p(\mathbf{x}_j \mid \mathcal{D}_{-j}).
\label{eq:cov-identity}
\end{equation}
The joint estimand scores a new cluster drawn from the population, covariates included; the conditional estimand scores that cluster's outcomes at a known covariate profile.
This is a different distinction from the conditional-versus-marginal treatment of the latent effect in Section~\ref{sec:setting}, where the quantity conditioned upon is unobserved and estimated from the very outcomes being scored, rather than observed and ancillary.

Two consequences follow from \eqref{eq:cov-identity}.
Firstly, candidate models carrying the same covariate block share the term $\log p(\mathbf{x}_j \mid \mathcal{D}_{-j})$, which cancels pointwise from their paired contrast, so the two formulations return identical $\elpd$ differences and identical paired standard errors.
Secondly, a joint $\elpd$ scores the pair $(\mathbf{y}_j, \mathbf{x}_j)$ and a conditional one
scores $\mathbf{y}_j$ alone, so the two are never comparable (two joint scores are comparable only when they cover the same variables).
Covariate selection under the joint formulation therefore deletes paths rather than variables, whereas under the conditional formulation candidate models may differ in their covariate sets.

Under the conditional formulation the covariate moments are data-dependent constants inside the likelihood, frozen at their full-sample values $(\bar{\mathbf{x}}, \mathbf{S}_x)$.
Deleting cluster $j$ would move them to $(\bar{\mathbf{x}}_{-j}, \mathbf{S}_{x,-j})$ and so disturb every remaining factor, rather than simply removing the $j$th, which is exactly what the reweighting identity of Proposition~\ref{prop:harmonic-mean} does not allow.
However, the following lemma proves that the conditional density does not involve the covariate moments $(\bm\mu_x, \bm\Sigma_{xx})$ at all, so the values at which they are frozen are immaterial.

\begin{lemma}[Invariance to the frozen block]
\label{lem:frozen}
Let the covariates enter the structural model through a coefficient matrix $\bm\Gamma$, so that $\bm\eta = \bm\alpha + \bm\Gamma \mathbf{x} + \mathbf{B} \bm\eta + \bm\zeta$.
Then the model-implied conditional density $p(\mathbf{y} \mid \mathbf{x}, \bm\vartheta)$ does not depend on the frozen covariate moments $(\bm\mu_x, \bm\Sigma_{xx})$.
\end{lemma}

\begin{proof}
Solving the structural equation and substituting into \eqref{eq:meas-single}, the model-implied moments of the enlarged vector $(\mathbf{y}^\top, \mathbf{x}^\top)^\top$ are
\begin{align*}
\bm\mu_y &= \bm\nu + \bm\Lambda (\mathbf{I} - \mathbf{B})^{-1} (\bm\alpha + \bm\Gamma \bm\mu_x), \\
\bm\Sigma_{yx} &= \bm\Lambda (\mathbf{I} - \mathbf{B})^{-1} \bm\Gamma \bm\Sigma_{xx}, \\
\bm\Sigma_{yy} &= \bm\Lambda (\mathbf{I} - \mathbf{B})^{-1} \bigl( \bm\Gamma \bm\Sigma_{xx} \bm\Gamma^\top + \bm\Psi \bigr) (\mathbf{I} - \mathbf{B})^{-\top} \bm\Lambda^\top + \bm\Theta .
\end{align*}
The conditional regression coefficient is $\bm\Sigma_{yx} \bm\Sigma_{xx}^{-1} = \bm\Lambda (\mathbf{I} - \mathbf{B})^{-1} \bm\Gamma$, in which $\bm\Sigma_{xx}$ has cancelled, so the conditional mean is
\begin{align*}
\bm\mu_y + \bm\Sigma_{yx} \bm\Sigma_{xx}^{-1} (\mathbf{x} - \bm\mu_x)
&= \bm\nu + \bm\Lambda (\mathbf{I} - \mathbf{B})^{-1} (\bm\alpha + \bm\Gamma \bm\mu_x) + \bm\Lambda (\mathbf{I} - \mathbf{B})^{-1} \bm\Gamma (\mathbf{x} - \bm\mu_x) \\
&= \bm\nu + \bm\Lambda (\mathbf{I} - \mathbf{B})^{-1} (\bm\alpha + \bm\Gamma \mathbf{x}),
\end{align*}
in which $\bm\mu_x$ has cancelled as well.
The conditional covariance is the Schur complement
\begin{align*}
\bm\Sigma_{yy} - \bm\Sigma_{yx} \bm\Sigma_{xx}^{-1} \bm\Sigma_{xy}
&= \bm\Lambda (\mathbf{I} - \mathbf{B})^{-1} \bigl( \bm\Gamma \bm\Sigma_{xx} \bm\Gamma^\top + \bm\Psi \bigr) (\mathbf{I} - \mathbf{B})^{-\top} \bm\Lambda^\top + \bm\Theta \\
&\qquad - \bm\Lambda (\mathbf{I} - \mathbf{B})^{-1} \bm\Gamma \bm\Sigma_{xx} \bm\Gamma^\top (\mathbf{I} - \mathbf{B})^{-\top} \bm\Lambda^\top,
\end{align*}
which reduces to $\bm\Lambda (\mathbf{I} - \mathbf{B})^{-1} \bm\Psi (\mathbf{I} - \mathbf{B})^{-\top} \bm\Lambda^\top + \bm\Theta$, the reduced-form covariance \eqref{eq:reduced-single} of the outcomes alone.
\end{proof}

In the two-level case the frozen block is the between-level marginal of the covariates, and since each $\bm\Sigma_l(\bm\vartheta)$ is the level-$l$ instance of the same reduced form, the cancellation applies within each level and hence to the cluster covariance $\mathbf{V}_j(\bm\vartheta)$ assembled from them.
The conditional cluster log-likelihood is therefore the joint one shifted by a $\bm\vartheta$-free constant, leaving $\mathbf{s}_j$ and $\mathbf{H}_j$ unchanged, so the conditional Taylor score inherits exactly the error of the joint one.
Everything above transfers to the single-row case (\loso) of Section~\ref{sec:loso-main} on replacing $(\mathbf{y}_j, \mathbf{x}_j)$ by $(\mathbf{y}_i, \mathbf{x}_i)$.

In summary, the two formulations are interchangeable wherever candidate models are compared over a common set of variables, as they are throughout this paper, and are otherwise separated by modelling choice, not scoring, as Table~\ref{tbl-cov-choice} sets out.

\begin{table}[tb]
\centering
\captionsetup{justification=justified}
\caption{Which formulation admits which use; joint is \texttt{fixed.x = FALSE} and conditional \texttt{fixed.x = TRUE}. On the first row, the case throughout this paper, the two do not merely both apply but agree exactly. A covariate decomposed across levels, or observed incompletely, must be given a distribution; a non-Gaussian one must not.}
\label{tbl-cov-choice}
\small
\begin{tabular}{lcc}
\toprule
 & Joint & Conditional \\
\midrule
Comparisons over a common covariate set & $\checkmark$ & $\checkmark$ \\
Comparisons over differing covariate sets & $\times$ & $\checkmark$ \\
Within- and between-cluster decomposition of a covariate & $\checkmark$ & $\times$ \\
Incompletely observed covariates & $\checkmark$ & $\times$ \\
Non-Gaussian covariates & $\times$ & $\checkmark$ \\
\bottomrule
\end{tabular}
\end{table}

\section{Proof of the Laplace Approximation of the Cluster CPO}
\label{app:laplace-proof}

\begin{proof}[Proof of Theorem~\ref{prop:laplace-2nd}]
Substituting the second-order expansion of $\ell_j$ into $I_j$ and the Gaussian density gives
\[
\hat{I}_j^{(2)} = \frac{e^{-\ell_j(\bm\vartheta^\ast)}}{(2\pi)^{d/2} |\bm\Omega|^{1/2}} \int \exp \Bigl(-\mathbf{s}_j^\top \bm\delta - \tfrac{1}{2}\, \bm\delta^\top \mathbf{A}_j \bm\delta\Bigr) d\bm\delta.
\]
Setting $\mathbf{H}_j = 0$ gives $\hat{I}_j^{(1)}$ with $\mathbf{A}_j$ replaced by $\bm\Omega^{-1}$, so the two truncations differ only in which precision matrix enters the Gaussian integral tilted by the linear form $-\mathbf{s}_j^\top \bm\delta$.
Diagonalising $\mathbf{A}_j = \mathbf{U} \operatorname{diag}(a_1, \dots, a_d) \mathbf{U}^\top$ with $\mathbf{U}$ orthogonal and substituting $\mathbf{z} := \mathbf{U}^\top \bm\delta$ factorises the integral into the $d$ scalar integrals $\int \exp\bigl(-t_i z_i - \tfrac{1}{2} a_i z_i^2\bigr) dz_i$ with $\mathbf{t} := \mathbf{U}^\top \mathbf{s}_j$, each finite exactly when $a_i > 0$ and divergent otherwise, whatever the value of $t_i$.
The integral therefore converges if and only if $\mathbf{A}_j \succ 0$, and congruence by $\bm\Omega^{1/2}$ turns this into the symmetric form on the right of \eqref{eq:existence-2nd}, since $-\bm\Omega^{1/2} \mathbf{H}_j \bm\Omega^{1/2}$ is symmetric with the same eigenvalues as $-\bm\Omega \mathbf{H}_j$.
At $\mathbf{H}_j = 0$, $\mathbf{A}_j = \bm\Omega^{-1} \succ 0$, so the integral is finite for every finite $\mathbf{s}_j$.

Completing the square, $\bm\delta^\top \mathbf{A}_j \bm\delta + 2 \mathbf{s}_j^\top \bm\delta = (\bm\delta + \mathbf{A}_j^{-1} \mathbf{s}_j)^\top \mathbf{A}_j (\bm\delta + \mathbf{A}_j^{-1} \mathbf{s}_j) - \mathbf{s}_j^\top \mathbf{A}_j^{-1} \mathbf{s}_j$, so the integral evaluates to $(2\pi)^{d/2} |\mathbf{A}_j|^{-1/2} \exp\bigl(\tfrac{1}{2} \mathbf{s}_j^\top \mathbf{A}_j^{-1} \mathbf{s}_j\bigr)$.
The determinant ratio is
\[
\frac{|\mathbf{A}_j|^{-1/2}}{|\bm\Omega|^{1/2}} = \bigl| \mathbf{A}_j \bm\Omega \bigr|^{-1/2} = \bigl|
\mathbf{I} + \bm\Omega \mathbf{H}_j \bigr|^{-1/2},
\]
the last equality using $\mathbf{A}_j \bm\Omega = \mathbf{I} + \mathbf{H}_j \bm\Omega$ and $|\mathbf{I} + \mathbf{H}_j \bm\Omega|
= |\mathbf{I} + \bm\Omega \mathbf{H}_j|$ (Sylvester's identity).
Hence, under \eqref{eq:existence-2nd}, $\hat{I}_j^{(2)} = \exp\bigl(-\ell_j(\bm\vartheta^\ast) + \tfrac{1}{2} \mathbf{s}_j^\top \mathbf{A}_j^{-1} \mathbf{s}_j\bigr) \cdot |\mathbf{I} + \bm\Omega \mathbf{H}_j|^{-1/2}$, where the determinant $|\mathbf{I} + \bm\Omega \mathbf{H}_j| = |\bm\Omega| \, |\mathbf{A}_j|$ is positive, so the log-determinant term is real, and $-\log \hat{I}_j^{(2)}$ gives \eqref{eq:laplace-2nd}.
At $\mathbf{H}_j = 0$, $\mathbf{A}_j = \bm\Omega^{-1}$ and $|\mathbf{I} + \bm\Omega \mathbf{H}_j| = 1$, so \eqref{eq:laplace-2nd} reduces to \eqref{eq:laplace-1st} exactly, as in the case of (i).
\end{proof}

\section{Marginal WAIC at First Order}
\label{app:waic}

The marginal WAIC targets the same leave-one-cluster-out quantity as the scores of Section~\ref{sec:laplace} \parencite{merkle2019bayesian}, and is usually computed from posterior draws.
Under a Gaussian density it needs none.
Both of its halves---the log pointwise predictive density $\lppd_j := \log \E_{\bm\vartheta \mid \mathbf{y}}\bigl[\, p(\mathbf{y}_j \mid \bm\vartheta) \,\bigr]$ and the penalty $p_{\mathrm{WAIC},j} := \Var_{\bm\vartheta \mid \mathbf{y}}\bigl[\, \ell_j(\bm\vartheta) \,\bigr]$ \parencite{watanabe2010asymptotic,vehtari2017practical}---are expectations of the same expansions of $\ell_j$ against the same $\mathcal{N}(\bm\vartheta^\ast, \bm\Omega)$, and so are functions of $(\mathbf{s}_j, \mathbf{H}_j, \bm\Omega)$ alone.
At first order they return the score of Theorem~\ref{prop:laplace-1st}(i) exactly.

\begin{corollary}[Marginal WAIC at first order]
\label{cor:waic-1st}
Under the Gaussian posterior approximation and the first-order expansion of Theorem~\ref{prop:laplace-1st}(i),
\begin{equation*}
\lppd_j = \ell_j(\bm\vartheta^\ast) + \tfrac{1}{2}\, \mathbf{s}_j^\top \bm\Omega \mathbf{s}_j,
\qquad
p_{\mathrm{WAIC},j} = \mathbf{s}_j^\top \bm\Omega \mathbf{s}_j,
\end{equation*}
so that $\lppd_j \!-\, p_{\mathrm{WAIC},j}$ is the first-order \loco{} score \eqref{eq:laplace-1st}.
\end{corollary}

\begin{proof}
Under the first-order expansion, $\ell_j(\bm\vartheta) = \ell_j(\bm\vartheta^\ast) + \mathbf{s}_j^\top \bm\delta$ is linear in $\bm\delta \sim \mathcal{N}(0, \bm\Omega)$, hence Gaussian with mean $\ell_j(\bm\vartheta^\ast)$ and variance $\mathbf{s}_j^\top \bm\Omega \mathbf{s}_j$.
Consequently, $\lppd_j$ is its cumulant generating function evaluated at one, $p_{\mathrm{WAIC},j}$ its variance, and their difference yields \eqref{eq:laplace-1st}.
\end{proof}

The resulting penalty $\tfrac{1}{2} p_{\mathrm{WAIC},j}$ in \eqref{eq:laplace-1st} is thus a Jensen-style one, falling hardest on clusters whose likelihood is most sensitive to $\bm\vartheta$.
Summed over clusters, $\sum_j p_{\mathrm{WAIC},j} = \operatorname{tr}(\bm\Omega \sum_j \mathbf{s}_j \mathbf{s}_j^\top)$, the outer-product form of the information, which agrees with the $p_{\mathrm{eff}}$ of \eqref{eq:peff-gap} to the same $O(J^{-1})$ under correct model specification.
This sum is also precisely what the \texttt{loo} package reports as $p_{\mathrm{loo}}$ \autocite{vehtari2017practical}, and in that case, estimated using PSIS.

\section{Proof of the Compatible-Prior Identity}
\label{app:compat}

\begin{proof}[Proof of Lemma~\ref{lem:submodel-id}]
By Bayes' theorem, $\pi(\bm\vartheta_{-a} \mid \mathbf{y}, \vartheta_a = 0) \propto p(\mathbf{y} \mid 0, \bm\vartheta_{-a})\,\pi(0, \bm\vartheta_{-a})$.
Factor the joint prior on the constraint slice as $\pi(0, \bm\vartheta_{-a}) = \pi_a(0)\, \pi(\bm\vartheta_{-a} \mid \vartheta_a = 0)$.
The marginal prior mass $\pi_a(0)$ is constant in $\bm\vartheta_{-a}$ and absorbs into the normalising constant.
The submodel refit posterior is $\pi_m(\bm\vartheta_{-a} \mid \mathbf{y}) \propto p(\mathbf{y} \mid 0, \bm\vartheta_{-a})\,\pi_m(\bm\vartheta_{-a})$ with the same likelihood, since ``$\vartheta_a$ fixed at zero'' and ``path $a$ absent'' specify the identical sampling model.
The two posteriors coincide iff $\pi_m(\bm\vartheta_{-a}) \propto \pi(\bm\vartheta_{-a} \mid \vartheta_a = 0)$, which equality of proper densities promotes to the hypothesis.
The general case $\mathbf{C}_m \bm\vartheta = 0$ follows by an invertible reparameterisation whose Jacobian cancels as $\pi_a(0)$ did.
\end{proof}

\section{Score and Hessian of a Restricted Submodel}
\label{app:trace}

The restricted score \eqref{eq:laplace-2nd-restricted} requires the free-subspace score $\mathbf{s}_{j,m,\mathrm{free}}$ and Hessian $\mathbf{H}_{j,m,\mathrm{free}}$ of the cluster log-likelihood at the restricted centre $\bm\vartheta^\ast_m$.
Taking the mean structure saturated so that structural paths do not enter $\bm\mu(\bm\vartheta)$, the score and Hessian are analytic in the model-implied covariances:
writing $\mathbf{W}_{j,m} = \mathbf{V}_j(\bm\vartheta^\ast_m)^{-1}$, $\mathbf{G}_a^{(j,m)} = \partial \mathbf{V}_j(\bm\vartheta^\ast_m)/\partial \vartheta_a$, $\mathbf{G}_{ag}^{(j,m)} = \partial^2 \mathbf{V}_j(\bm\vartheta^\ast_m)/\partial \vartheta_a\, \partial \vartheta_g$, and $\tilde{\mathbf{y}}_j = \mathbf{y}_j - \mathbf{1}_{n_j} \otimes \bm\mu(\bm\vartheta^\ast_m)$,
\begin{align}
(\mathbf{s}_{j,m,\mathrm{free}})_a 
&= -\tfrac{1}{2}\,\tr \bigl(\mathbf{W}_{j,m} \mathbf{G}_a^{(j,m)} (\mathbf{I} - \mathbf{W}_{j,m}\,\tilde{\mathbf{y}}_j \tilde{\mathbf{y}}_j^\top)\bigr), \label{eq:score-restricted-free} \\
(\mathbf{H}_{j,m,\mathrm{free}})_{ag} 
&= -\tfrac{1}{2}\,\tr \bigl(\mathbf{W}_{j,m} \mathbf{G}_a^{(j,m)} \mathbf{W}_{j,m} \mathbf{G}_g^{(j,m)} (2\mathbf{W}_{j,m}\,\tilde{\mathbf{y}}_j \tilde{\mathbf{y}}_j^\top - \mathbf{I})\bigr) \nonumber \\
&\quad \quad -\tfrac{1}{2}\,\tr \bigl(\mathbf{W}_{j,m} \mathbf{G}_{ag}^{(j,m)} (\mathbf{I} - \mathbf{W}_{j,m}\,\tilde{\mathbf{y}}_j \tilde{\mathbf{y}}_j^\top)\bigr). \label{eq:hessian-restricted-free}
\end{align}
The derivatives $\mathbf{G}_a^{(j,m)}, \mathbf{G}_{ag}^{(j,m)}$ follow by chain rule from $\bm\Sigma_{\mathrm{w}}(\bm\vartheta), \bm\Sigma_{\mathrm{b}}(\bm\vartheta)$ and are obtained once per candidate.

Saturation of the mean structure is an assumption of these formulae, not of the construction.
It holds for every model fitted in this paper: the mean structure is estimated throughout (Section~\ref{sec:setting}) with free indicator intercepts $\bm\nu$ and latent intercepts fixed at zero, the \texttt{lavaan}-family identification default under \texttt{meanstructure = TRUE} used by the fits of Sections~\ref{sec:pisa} and~\ref{sec:midus}, so that $\bm\mu(\bm\vartheta) = \bm\nu$ and the structural matrices enter the likelihood only through the covariances.
A mean structure through which structural paths do act contributes the standard mean-derivative terms to \eqref{eq:score-restricted-free}--\eqref{eq:hessian-restricted-free} and changes nothing else in the construction.

\end{document}